\documentclass[11pt]{article}

\usepackage{fullpage}

\usepackage{makecell}
\usepackage{array}
\usepackage{comment}
\usepackage{epsf}
\usepackage{subcaption}
\usepackage{graphicx}
\usepackage{bbm}
\usepackage{color}
\usepackage{nicefrac}
\usepackage{tcolorbox}
\usepackage{mathtools}
\usepackage{amsfonts}
\usepackage{amsthm}
\usepackage{amsmath, bm}
\usepackage{amssymb}
\usepackage{textcomp}
\usepackage{lipsum}
\usepackage{xcolor}
\usepackage{graphicx}
\usepackage{wrapfig}
\usepackage{mathtools}
\usepackage{enumitem}
\usepackage{xspace}

\usepackage{algorithm}
\usepackage{algpseudocode}
\counterwithin{algorithm}{section}

\usepackage{tikz}
\usetikzlibrary{shapes,snakes}
\usetikzlibrary{arrows.meta}
\usetikzlibrary{decorations.pathmorphing}
\usetikzlibrary{patterns}
\usetikzlibrary{calc}

\usepackage{amsbsy}

\usepackage[linktocpage=true]{hyperref}
\definecolor{OliveGreen}{HTML}{3C8031}
\hypersetup{
colorlinks=true,
urlcolor=blue,
linkcolor=blue,
citecolor=OliveGreen,
}

\usepackage{xurl}
\hypersetup{breaklinks=true}

\newenvironment{fminipage}%
  {\end{minipage}\end{Sbox}\fbox{\TheSbox}}

\newtheorem{theorem}{Theorem}[section]

\newtheorem{proposition}[theorem]{Proposition}
\newtheorem*{proposition*}{Proposition}
\newtheorem{lemma}[theorem]{Lemma}

\newtheorem{observation}[theorem]{Observation}

\newtheorem{fact}[theorem]{Fact}
\newtheorem{question}{Question}

\theoremstyle{definition}
\newtheorem{definition}{Definition}[section]

\makeatletter
\newenvironment{breakablealgorithm}
  {
   \begin{center}
     \refstepcounter{algorithm}
     \hrule height.8pt depth0pt \kern2pt
     \renewcommand{\caption}[2][\relax]{
       {\raggedright\textbf{\ALG@name~\thealgorithm} ##2\par}%
       \ifx\relax##1\relax 
         \addcontentsline{loa}{algorithm}{\protect\numberline{\thealgorithm}##2}%
       \else 
         \addcontentsline{loa}{algorithm}{\protect\numberline{\thealgorithm}##1}%
       \fi
       \kern2pt\hrule\kern2pt
     }
  }{
     \kern2pt\hrule\relax
   \end{center}
  }
\makeatother

\usepackage[backend=biber, style=alphabetic, sorting=nyt, maxnames=100,backref=true, firstinits=true, sortcites = true]{biblatex}
\newlength{\widebarargwidth}
\newlength{\widebarargheight}
\newlength{\widebarargdepth}

\makeatletter
\long\def\@makecaption#1#2{
       \vskip 0.8ex
       \setbox\@tempboxa\hbox{\small {\bf #1:} #2}
       \parindent 1.5em 
       \dimen0=\hsize
       \advance\dimen0 by -3em
       \ifdim \wd\@tempboxa >\dimen0
               \hbox to \hsize{
                       \parindent 0em
                       \hfil 
                       \parbox{\dimen0}{\def\baselinestretch{0.96}\small
                               {\bf #1.} #2
                               } 
                       \hfil}
       \else \hbox to \hsize{\hfil \box\@tempboxa \hfil}
       \fi
       }
\makeatother

\long\def\comment#1{}

\newcommand{\bbone}[1]{\mathbbm{1}\{#1\}}
\newcommand{\set}[1]{\{#1\}}
\newcommand{\defeq}{\coloneqq}
\newcommand{\N}{\mathbb{N}}

\newcommand{\E}{\mathbb{E}}

\newcommand{\poly}{\mathsf{poly}}

\newcommand{\vend}{\mathsf{vEnd}}
\newcommand{\vstart}{\mathsf{vStart}}

\newcommand{\Pivot}{\mathsf{Pivot}}
\newcommand{\length}{\mathsf{length}}
\renewcommand{\epsilon}{\varepsilon}
\newcommand{\eps}{\epsilon}

\newcommand{\dom}{\mathsf{dom}}
\newcommand{\0}{\varnothing}
\newcommand{\bemph}[1]{{\normalfont#1}} 
\newcommand{\ep}[1]{\bemph{(}#1\bemph{)}} 

\newcommand{\emphdef}[1]{\textbf{\textit{{#1}}}}
\newcommand{\emphd}[1]{\emphdef{#1}}
\newcommand{\blank}{\mathsf{blank}}

\title{A Simple Algorithm for Near-Vizing Edge-Coloring in\\Near-Linear Time}

\usepackage{authblk}

\author[1]{Abhishek Dhawan\thanks{Email: abhishek.dhawan@math.gatech.edu}}

\affil[1]{School of Mathematics, Georgia Institute of Technology}

\date{}

\begin{document}

\maketitle

\begin{abstract}
    We present a simple $(1+\eps)\Delta$-edge-coloring algorithm for graphs of maximum degree $\Delta = \Omega(\log n / \eps)$ with running time $O\left(m\,\log^3 n/\eps^3\right)$. Our algorithm improves upon that of [Duan, He, and Zhang; SODA19], which was the first near-linear time algorithm for this problem. While our results are weaker than the current state-of-the-art, our approach is significantly simpler, both in terms of analysis as well as implementation, and may be of practical interest.
\end{abstract}

\vspace{10pt}
\tableofcontents
\newpage

\section{Introduction}\label{section:intro}

\subsection{Background and Results}\label{subsec:background}

All graphs considered in this paper are finite, undirected, and simple.
Let $G = (V, E)$ be a graph satisfying $|V| = n$, $|E| = m$, and $\Delta(G) = \Delta$, where $\Delta(G)$ is the maximum degree of a vertex in $G$.
For $q \in \N$, a proper $q$-edge-coloring of $G$ is a function $\phi\,:\, E \to [q]$ (where $[q] = \set{1, \ldots, q}$) such that $\phi(e) \neq \phi(f)$ whenever $e$ and $f$ share an endpoint.
The chromatic index of $G$, denoted $\chi'(G)$, is the minimum value $q$ for which $G$ admits a proper $q$-edge-coloring.
Edge-coloring is one of the most fundamental and well-studied problems in graph theory, having both great theoretical significance and a wide array of applications.
Virtually every graph theory textbook contains a chapter on edge-coloring, see, e.g., \cites[\S17]{BondyMurty}[\S5.3]{Diestel}.

It is easy to see that $\chi'(G) \geq \Delta$ as each edge incident to a vertex of maximum degree must receive a different color.
Vizing's celebrated result shows that the chromatic index is at most one away from this trivial lower bound, i.e., $\chi'(G) \leq \Delta + 1$ \cite{Vizing}.
(See \cite[\S{}A.1]{EdgeColoringMonograph} for an English translation of the paper and \cites[\S17.2]{BondyMurty}[\S5.3]{Diestel} for modern presentations.)
In general, it is NP-hard to determine $\chi'(G)$ (even for $\Delta = 3$) \cite{Holyer}.
Therefore, it is natural to design algorithms for $q$-edge-coloring where $q \geq \Delta + 1$.

When $q = \Delta + 1$, we refer to such colorings as \emphd{Vizing edge-colorings}.
The original proof of Vizing's theorem is constructive and yields an $O(mn)$-time algorithm, simplifications of which appeared in \cite{Bollobas, RD, MG}.
Gabow, Nishizeki, Kariv, Leven, and Terada designed two recursive algorithms, which run in $O(m\sqrt{n\log n})$ and $O(\Delta\,m\log n)$ time, respectively \cite{GNKLT}.
There was no improvement for 34 years until Sinnamon described an algorithm with running time $O(m\sqrt{n})$ \cite{Sinnamon}, and more recently Bhattacharya, Carmon, Costa, Solomon, and Zhang designed an $\tilde O(mn^{1/3})$-time\footnote{Here, and in what follows, $\tilde O(x) = O(x\,\poly(\log x))$} randomized algorithm \cite{bhattacharya2024faster}.

While the above summarizes the history of the state-of-the-art for Vizing edge-coloring of general graphs, there has been additional progress for restricted classes of graphs.
Notably, Cole, Ost, and Schirra designed an $O(m\log \Delta)$-time algorithm for $\Delta$-edge-coloring bipartite graphs \cite{cole2001edge};
Bernshteyn and the author considered graphs of bounded maximum degree, describing a linear-time algorithm for Vizing edge-coloring in this setting, i.e., $O(m)$ for $\Delta = O(1)$ \cite{fastEdgeColoring};
Bhattacharya, Carmon, Costa, Solomon, and Zhang designed an $\tilde O(m)$-time algorithm for graphs of bounded arboricity \cite{bhattacharya2023density}.

As exhibited by the results mentioned in the previous paragraph, we may obtain faster algorithms by considering restricted classes of graphs.
Another option is to use larger palettes, i.e., more colors.
A \emphd{near-Vizing edge-coloring} is a proper edge-coloring with palettes of size $(1+\eps)\Delta$ for $\eps > 0$.
This is the main focus of this paper.
We summarize the history of near-Vizing edge-coloring in Table~\ref{table:history}.

\begin{table}[htb!]
    \centering
    \begin{tabular}{| c | c | c | >{\centering\arraybackslash\scriptsize}m{0.25\textwidth}|}
        \hline
        \thead{\textbf{\#{} Colors}} & \thead{\textbf{Runtime}} & \thead{\textbf{Restrictions on $\Delta,\,n,\,\eps$}} & \thead{\textbf{References}} \\\hline\hline
        $(1+\epsilon)\Delta$ & $O\left(m\,\log^6 n/\epsilon^2\right)^\star$ & $\Delta = \Omega\left(\log n / \eps\right)$ & {Duan--He--Zhang \cite{duan2019dynamic}} \\[2pt]\hline
         $(1+\epsilon)\Delta$ & $O\left(m\,\log n/\epsilon\right)^\star$ & No restrictions & {Elkin--Khuzman \cite{elkin2024deterministic}} \\[2pt]\hline
         $(1+\epsilon)\Delta$ & $O\left(\max \{m/\epsilon^{18}, \ m \,\log\Delta\}\right)^\star$ & $\eps \geq \max\left\{\frac{\log \log n}{\log n}, \frac{1}{\Delta}\right\}$ & {Elkin--Khuzman \cite{elkin2024deterministic}} \\[4pt]\hline
         $(1+\epsilon)\Delta$ & $O\left(m \,\log(1/\epsilon)/\epsilon^2\right)^\star$ &  $\Delta \geq \left(\log n / \epsilon\right)^{\poly\left(1/\epsilon\right)}$ & {Bhattacharya--Costa--Panski--Solomon \cite{bhattacharya2024nibbling}} \\[2pt]\hline
         $(1+\epsilon)\Delta$ & $O\left(m \,\log(1/\epsilon)\right)^\star$ in expectation &  $\Delta = \Omega\left(\log n / \eps\right)$ & {Assadi \cite{assadi2024faster}} \\[2pt]\hline
         $(1+\epsilon)\Delta$ & $O\left(m\log (1/\eps)/\eps^4\right)^\star$ & No restrictions & {Bernshteyn--AD \cite{bernshteyn2024linear}} \\[2pt]\hline
    \end{tabular}
    
    \vspace{7pt}
    \caption{A brief survey of near-Vizing edge-coloring algorithms ($\star$ indicates randomized). The stated runtime of randomized algorithms is attained with high probability, unless explicitly indicated otherwise.}\label{table:history}
\end{table}

In our work, we design a randomized algorithm for dense graphs, i.e., $\Delta = \Omega(\log n / \eps)$.
While a number of results in Table~\ref{table:history} outperform our algorithm, we remark that our approach is significantly simpler, both in terms of analysis as well as implementation, and may be of practical interest.
We discuss this further below.
Let us first state our result.

\begin{theorem}\label{theo:main_theo}
    Let $\epsilon \in (0,1)$ be arbitrary, and let 
    $G$ be an $n$-vertex graph with $m$ edges having maximum degree $\Delta \geq 500\log n / \epsilon$.
    There is a randomized sequential algorithm that finds a proper $(1+\epsilon)\Delta$-edge-coloring of $G$ in time $O\left(m\log^3n/\epsilon^3\right)$ with probability at least $1 - 1/\poly(n)$.
\end{theorem}

Note that our result above asymptotically improves upon that of Duan, He, and Zhang stated in Table~\ref{table:history}.
Additionally, for $\Delta = O(\log n /\eps)$, there are more efficient algorithms for Vizing edge-coloring (see \cite[\S4]{fastEdgeColoring} for a simple $O(\Delta^3\,m\log n)$-time algorithm) and so the dense regime is really the only regime of interest.

The remainder of this introduction is structured as follows: in \S\ref{subsection: comparison}, we discuss the approaches of the results stated so far in more detail in order to justify the simplicity of our approach; in \S\ref{subsection: overview}, we provide an informal overview of our algorithm and a proof sketch of the analysis; and in \S\ref{subsection: conclusion}, we outline potential future directions.

\subsection{Discussion of Prior Work}\label{subsection: comparison}

Throughout the rest of the paper, we fix a graph $G$ with $n$ vertices, $m$ edges, and maximum degree $\Delta$, and we write $V \defeq V(G)$ and $E \defeq E(G)$.
We fix $\epsilon \in (0,1)$ such that $\Delta = \Omega\left(\log n/\epsilon\right)$ (where the hidden constant is assumed to be sufficiently large) and set $q \defeq (1+\epsilon)\Delta$. 
Without loss of generality, we may assume that $q$ is an integer. 
We call a function $\phi \colon E\to [q]\cup \{\blank\}$ a \emphd{partial $q$-edge-coloring} (or simply a \emphd{partial coloring}) of $G$. Here $\phi(e) = \blank$ indicates that the edge $e$ is uncolored. As usual, $\dom(\phi)$ denotes the \emphd{domain} of $\phi$, i.e., the set of all colored edges.

A standard approach toward edge-coloring is to construct so-called \emphd{augmenting subgraphs}.
In fact, nearly all of the algorithms mentioned so far (with the exception of \cite{assadi2024faster, bhattacharya2024nibbling}) employ this approach.
The idea, in a nutshell, is to extend a partial coloring to include an uncolored edge by modifying the colors of ``few'' colored edges.
Formally, we define augmenting subgraphs as follows (this definition is taken from \cite[Definition~1.4]{bernshteyn2024linear}):

\begin{definition}[Augmenting subgraphs]\label{defn:aug}
    Let $\phi \colon E \to [q] \cup \set{\blank}$ be a proper partial $q$-edge-coloring with domain $\dom(\phi) \subset E$. A subgraph $H \subseteq G$ is \emphd{$e$-augmenting} for an uncolored edge $e \in E \setminus \dom(\phi)$ if $e \in E(H)$ and there is a proper partial coloring $\phi'$ with $\dom(\phi') = \dom(\phi) \cup \set{e}$ that agrees with $\phi$ on the edges that are not in $E(H)$; in other words, by only modifying the colors of the edges of $H$, it is possible to add $e$ to the set of colored edges. We refer to such a modification operation as \emphd{augmenting} $\phi$ using $H$.
\end{definition}

This yields an algorithmic framework for edge-coloring.
Namely, proceed in an iterative fashion and at each iteration, first construct an augmenting subgraph $H$ with respect to the current partial coloring $\phi$, and then augment $\phi$ using $H$.
This is the essence of the idea described earlier.
See Algorithm~\ref{temp:seq} for an outline of this framework.

{
\floatname{algorithm}{Algorithm Template}
\begin{algorithm}[htb!]\small
    \caption{A $q$-edge-coloring algorithm}\label{temp:seq}
    \begin{flushleft}
        \textbf{Input}: A graph $G = (V,E)$ of maximum degree $\Delta$. \\
        \textbf{Output}: A proper $q$-edge-coloring of $G$.
    \end{flushleft}
    \begin{algorithmic}[1]
        \State $\phi \gets$ the empty coloring
        \While{there are uncolored edges}
            \State\label{step:choose_S} Pick an uncolored edge $e$.
            \State\label{step:construct_H} Find an $e$-augmenting subgraph $H$.
            \State Augment $\phi$ using $H$ (thus adding $e$ to the set of colored edges).
        \EndWhile
        \State \Return $\phi$
    \end{algorithmic}
\end{algorithm}
}

As mentioned earlier, nearly all of the algorithms in Table~\ref{table:history} employ this template (with some modifications).
In Vizing's original proof, he describes how to construct an $e$-augmenting subgraph (although he did not use this terminology) for any uncolored edge $e$ consisting of a \emphd{fan}---i.e., a set of edges incident to a common vertex---and an \emphd{alternating path}---i.e., a path whose edge colors form the sequence $\alpha$, $\beta$, $\alpha$, $\beta$, \ldots{} for some $\alpha$, $\beta \in [q]$; see Fig.~\ref{fig:vizing} for an illustration.
Such an augmenting subgraph, which we call a \emphd{Vizing chain} and denote $(F, P)$ for the fan $F$ and path $P$, can be constructed and augmented in time proportional to the lengths of $F$ and $P$, i.e., the number of edges contained in the Vizing chain.
As the length of a fan is trivially at most $\Delta$ and that of an alternating path is at most $n$, this yields the $O(m(\Delta + n)) = O(mn)$ runtime of the algorithms in \cite{Bollobas, RD, MG}.
There are now two challenges to overcome: long fans and long paths.

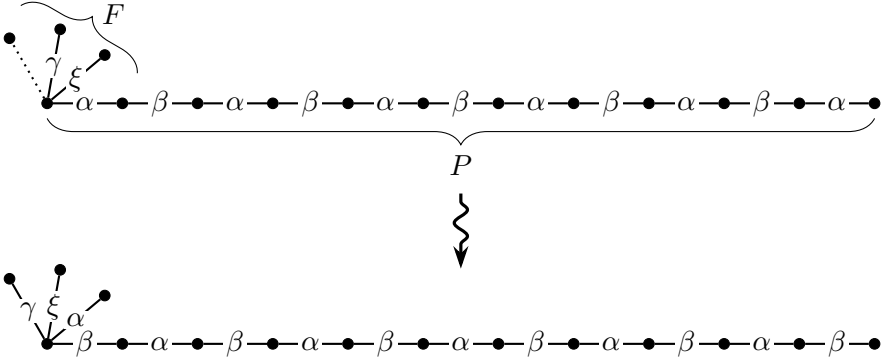
\begin{figure}[htb!]
    \centering
        \begin{tikzpicture}
            \node[circle,fill=black,draw,inner sep=0pt,minimum size=4pt] (a) at (0,0) {};
        	\path (a) ++(0:1) node[circle,fill=black,draw,inner sep=0pt,minimum size=4pt] (b) {};
        	\path (b) ++(0:1) node[circle,fill=black,draw,inner sep=0pt,minimum size=4pt] (c) {};
                \path (c) ++(0:1) node[circle,fill=black,draw,inner sep=0pt,minimum size=4pt] (d) {};
                \path (d) ++(0:1) node[circle,fill=black,draw,inner sep=0pt,minimum size=4pt] (e) {};
                \path (e) ++(0:1) node[circle,fill=black,draw,inner sep=0pt,minimum size=4pt] (l) {};
                \path (l) ++(0:1) node[circle,fill=black,draw,inner sep=0pt,minimum size=4pt] (m) {};
                \path (m) ++(0:1) node[circle,fill=black,draw,inner sep=0pt,minimum size=4pt] (n) {};
                \path (n) ++(0:1) node[circle,fill=black,draw,inner sep=0pt,minimum size=4pt] (o) {};
                \path (o) ++(0:1) node[circle,fill=black,draw,inner sep=0pt,minimum size=4pt] (f) {};
                \path (f) ++(0:1) node[circle,fill=black,draw,inner sep=0pt,minimum size=4pt] (g) {};
                \path (g) ++(0:1) node[circle,fill=black,draw,inner sep=0pt,minimum size=4pt] (h) {};

                \path (a) ++(40:1) node[circle,fill=black,draw,inner sep=0pt,minimum size=4pt] (i) {};
                \path (a) ++(80:1) node[circle,fill=black,draw,inner sep=0pt,minimum size=4pt] (j) {};
                \path (a) ++(120:1) node[circle,fill=black,draw,inner sep=0pt,minimum size=4pt] (k) {};

                \draw[thick] (a) to node[midway,inner sep=1pt,outer sep=1pt,minimum size=4pt,fill=white] {$\alpha$} (i) (a) to node[midway,inner sep=1pt,outer sep=1pt,minimum size=4pt,fill=white] {$\xi$} (j) (a) to node[midway,inner sep=1pt,outer sep=1pt,minimum size=4pt,fill=white] {$\gamma$} (k);
         
        	\draw[thick] (a) to node[midway,inner sep=1pt,outer sep=1pt,minimum size=4pt,fill=white] {$\beta$} (b) to node[midway,inner sep=1pt,outer sep=1pt,minimum size=4pt,fill=white] {$\alpha$} (c) to node[midway,inner sep=1pt,outer sep=1pt,minimum size=4pt,fill=white] {$\beta$} (d) to node[midway,inner sep=1pt,outer sep=1pt,minimum size=4pt,fill=white] {$\alpha$} (e) to node[midway,inner sep=1pt,outer sep=1pt,minimum size=4pt,fill=white] {$\beta$} (l) to node[midway,inner sep=1pt,outer sep=1pt,minimum size=4pt,fill=white] {$\alpha$} (m) to node[midway,inner sep=1pt,outer sep=1pt,minimum size=4pt,fill=white] {$\beta$} (n) to node[midway,inner sep=1pt,outer sep=1pt,minimum size=4pt,fill=white] {$\alpha$} (o) to node[midway,inner sep=1pt,outer sep=1pt,minimum size=4pt,fill=white] {$\beta$} (f) to node[midway,inner sep=1pt,outer sep=1pt,minimum size=4pt,fill=white] {$\alpha$} (g) to node[midway,inner sep=1pt,outer sep=1pt,minimum size=4pt,fill=white] {$\beta$} (h);

        \begin{scope}[yshift=3.2cm]
            \node[circle,fill=black,draw,inner sep=0pt,minimum size=4pt] (a) at (0,0) {};
        	\path (a) ++(0:1) node[circle,fill=black,draw,inner sep=0pt,minimum size=4pt] (b) {};
        	\path (b) ++(0:1) node[circle,fill=black,draw,inner sep=0pt,minimum size=4pt] (c) {};
                \path (c) ++(0:1) node[circle,fill=black,draw,inner sep=0pt,minimum size=4pt] (d) {};
                \path (d) ++(0:1) node[circle,fill=black,draw,inner sep=0pt,minimum size=4pt] (e) {};
                \path (e) ++(0:1) node[circle,fill=black,draw,inner sep=0pt,minimum size=4pt] (l) {};
                \path (l) ++(0:1) node[circle,fill=black,draw,inner sep=0pt,minimum size=4pt] (m) {};
                \path (m) ++(0:1) node[circle,fill=black,draw,inner sep=0pt,minimum size=4pt] (n) {};
                \path (n) ++(0:1) node[circle,fill=black,draw,inner sep=0pt,minimum size=4pt] (o) {};
                \path (o) ++(0:1) node[circle,fill=black,draw,inner sep=0pt,minimum size=4pt] (f) {};
                \path (f) ++(0:1) node[circle,fill=black,draw,inner sep=0pt,minimum size=4pt] (g) {};
                \path (g) ++(0:1) node[circle,fill=black,draw,inner sep=0pt,minimum size=4pt] (h) {};

                \path (a) ++(40:1) node[circle,fill=black,draw,inner sep=0pt,minimum size=4pt] (i) {};
                \path (a) ++(80:1) node[circle,fill=black,draw,inner sep=0pt,minimum size=4pt] (j) {};
                \path (a) ++(120:1) node[circle,fill=black,draw,inner sep=0pt,minimum size=4pt] (k) {};

                \draw[thick] (a) to node[midway,inner sep=1pt,outer sep=1pt,minimum size=4pt,fill=white] {$\xi$} (i) (a) to node[midway,inner sep=1pt,outer sep=1pt,minimum size=4pt,fill=white] {$\gamma$} (j);

                \draw[thick, dotted] (a) -- (k);
         
        	\draw[thick] (a) to node[midway,inner sep=1pt,outer sep=1pt,minimum size=4pt,fill=white] {$\alpha$} (b) to node[midway,inner sep=1pt,outer sep=1pt,minimum size=4pt,fill=white] {$\beta$} (c) to node[midway,inner sep=1pt,outer sep=1pt,minimum size=4pt,fill=white] {$\alpha$} (d) to node[midway,inner sep=1pt,outer sep=1pt,minimum size=4pt,fill=white] {$\beta$} (e) to node[midway,inner sep=1pt,outer sep=1pt,minimum size=4pt,fill=white] {$\alpha$} (l) to node[midway,inner sep=1pt,outer sep=1pt,minimum size=4pt,fill=white] {$\beta$} (m) to node[midway,inner sep=1pt,outer sep=1pt,minimum size=4pt,fill=white] {$\alpha$} (n) to node[midway,inner sep=1pt,outer sep=1pt,minimum size=4pt,fill=white] {$\beta$} (o) to node[midway,inner sep=1pt,outer sep=1pt,minimum size=4pt,fill=white] {$\alpha$} (f) to node[midway,inner sep=1pt,outer sep=1pt,minimum size=4pt,fill=white] {$\beta$} (g) to node[midway,inner sep=1pt,outer sep=1pt,minimum size=4pt,fill=white] {$\alpha$} (h);

            \draw[decoration={brace,amplitude=10pt,mirror}, decorate] (0, -0.2) -- node [midway,below,xshift=0pt,yshift=-10pt] {$P$} (11,-0.2);
            \draw[decoration={brace,amplitude=10pt},decorate] (-0.35,1.3) -- node [midway,above,yshift=2pt,xshift=13pt] {$F$} (1.2, 0.4);
            
        \end{scope}

        \begin{scope}[yshift=2.5cm]
            \draw[-{Stealth[length=3mm,width=2mm]},very thick,decoration = {snake,pre length=3pt,post length=7pt,},decorate] (5.5,-0.5) -- (5.5,-1.5);
        \end{scope}
        	
        \end{tikzpicture}
    \caption{The process of augmenting a Vizing chain $(F, P)$.}
    \label{fig:vizing}
\end{figure}

In \cite{duan2019dynamic}, Duan, He, and Zhang made the following algorithmic observation regarding the proof of Vizing's theorem: 

\begin{observation}\label{obs: vizing}
    Let $G$ be a graph of maximum degree $\Delta$ and let $\phi$ be a proper partial $q$-edge-coloring of $G$ for $q \geq \Delta + 1$.
    Let $C \subseteq [q]$ be such that for each $v \in V(G)$, there is at least $1$ color in $C$ that is not used on any of its incident edges.
    Then there is an algorithm which constructs an $e$-augmenting Vizing chain $(F, P)$ for any uncolored edge $e$ in time $O(|C|^2 + \length(P))$ such that the colored edges contained in $(F, P)$ are colored with colors from $C$.
\end{observation}

Note that setting $C = [\Delta + 1]$ above yields an algorithm for Vizing edge-coloring.
Furthermore, as $\phi$ is proper, it must be the case that $\length(F) \leq |C|$.
With this in mind, Duan, He, and Zhang describe the following algorithm to construct an $e$-augmenting Vizing chain at Step~\ref{step:construct_H}:
\begin{itemize}
    \item sample a palette $C \subseteq [q]$ satisfying $|C| = O(\log n / \eps)$, 
    \item construct an $e$-augmenting Vizing chain $(F, P)$ consisting of edges colored with colors from $C$ (note that $\length(F) = O(\log n/\eps)$),
    \item if the path $P$ is ``long'', try again.
\end{itemize}
They show that $C$ satisfies the conditions of Observation~\ref{obs: vizing} with high probability.
Furthermore, for an appropriate notion of ``long'', they are able to show that the above procedure succeeds within $\poly(\log n, 1/\eps)$ attempts with high probability.

The algorithms of Elkin and Khuzman \cite{elkin2024deterministic} and Bernshteyn and the author \cite{bernshteyn2024linear} build upon earlier work of Bernshteyn and the author \cite{fastEdgeColoring}.
As mentioned in \S\ref{subsec:background}, it provides (among other things) a $(\Delta + 1)$-edge-coloring algorithm with running time $O(m)$ for $\Delta = O(1)$.
In \cite{elkin2024deterministic}, Elkin and Khuzman use it as a subroutine in order to construct a proper $(1+\epsilon)\Delta$-edge-coloring. 
Roughly, their idea is to recursively split the edges of $G$ into subgraphs of progressively smaller maximum degree, eventually reducing the degree to a constant, and then run the algorithm from \cite{fastEdgeColoring} (which follows the template of Algorithm~\ref{temp:seq}) on each constant-degree subgraph separately. 
The approach of \cite{bernshteyn2024linear} is to instead implement a variant of their \textsf{Multi-Step Vizing Algorithm} from \cite{fastEdgeColoring} on $G$ directly and make use of the larger set of available colors ($(1+\epsilon)\Delta$ rather than $\Delta + 1$) in the analysis.
Given a partial coloring $\phi$ and an uncolored edge $e$, it outputs an $e$-augmenting subgraph $H$ of a special form, called a \emphd{multi-step Vizing chain}, which as the name suggests, consists of multiple Vizing chains glued together in an intricate fashion.

We conclude with a discussion of the two algorithms which do not follow the template of Algorithm~\ref{temp:seq}.
In \cite{bhattacharya2024nibbling}, the authors employ a powerful combinatorial tool known as the ``R\"odl nibble method.''
They describe a two-stage coloring procedure: in Stage 1, color ``most'' of the edges of $G$ using the nibble method; in Stage 2, color the remaining edges using a ``folklore'' linear-time $(2+\eps)\Delta$-edge-coloring algorithm (see \S\ref{subsec:stage two} for a description of this algorithm).
In order for the algorithm to succeed, one must ensure the uncolored edges after Stage 1 induce a graph of maximum degree $O(\eps\Delta)$.
The authors show this holds with high probability through a delicate argument involving a number of probabilistic and structural tools.
In \cite{assadi2024faster}, Assadi also describes a two-stage coloring procedure.
To use their terminology, they design a procedure to construct $\Delta$ disjoint \emphd{fair matchings} in linear time.
(The precise definition of ``fair'' is rather technical and unrelated to our work and so we omit it here.
Suffice it to say that the fairness plays a key role in ensuring the uncolored edges after Stage 1 induce a graph of small maximum degree.)
The procedure and the analysis are technical, involving a random walk over an auxiliary directed graph as well as the use of a number of probabilistic and structural tools.

Note that each of the ideas mentioned in this section (multi-step Vizing chains, the nibble method, random walks, etc.) are complex and involve intricate arguments.
We emphasize that although our results are weaker, the approach is significantly simpler, both in the ideas and in the analysis.

\subsection{Informal Overview}\label{subsection: overview}

In this section, we will provide an informal overview of our algorithm as well as a proof sketch of the analysis.
For simplicity, we will describe how to compute a $(1 + O(\eps))\Delta$-edge-coloring (the argument for $q = (1+\eps)\Delta$ follows by re-parameterizing $\eps \gets \Theta(\eps)$).
We design a two-stage coloring procedure, where the first stage constructs a proper partial $(1 + O(\eps))\Delta$-edge-coloring, and the second stage completes the coloring by employing the folklore $(2+\eps)\Delta$-edge-coloring algorithm (see Algorithm~\ref{alg:greedy}).

Following the template of Algorithm~\ref{temp:seq}, our algorithm for Stage 1 differs from that of Duan, He, and Zhang \cite{duan2019dynamic} in two ways: the implementation of Steps~\ref{step:choose_S} and \ref{step:construct_H}.
We first discuss our implementation of Step~\ref{step:construct_H}.
We construct a Vizing chain in an identical fashion to that of \cite{duan2019dynamic}, i.e., sample a palette $C\subseteq [q]$ of $O(\log n / \eps)$ colors and construct an $e$-augmenting Vizing chain $(F, P)$ consisting of edges colored with colors from $C$.
For an appropriate parameter $\ell \in \N$, we say $P$ is \emphd{long} if $\length(P) \geq \ell$ and \emphd{short} otherwise.
As mentioned earlier, the algorithm of \cite{duan2019dynamic} augments the current coloring using $(F, P)$ if $P$ is short, and repeats the above process if $P$ is long.
Rather than repeating the process if $P$ is long, we instead sample one of the first $\ell$ edges on $P$ uniformly at random and \emphd{flag} it (flagged edges are colored during Stage 2).
We then augment the \emphd{initial segment} of the Vizing chain ending just before this edge.
(A similar idea was employed by Su and Vu in their distributed algorithm for $(\Delta + 2)$-edge-coloring \cite{su2019towards}.)
We illustrate this procedure in Fig.~\ref{fig:main_change} below:

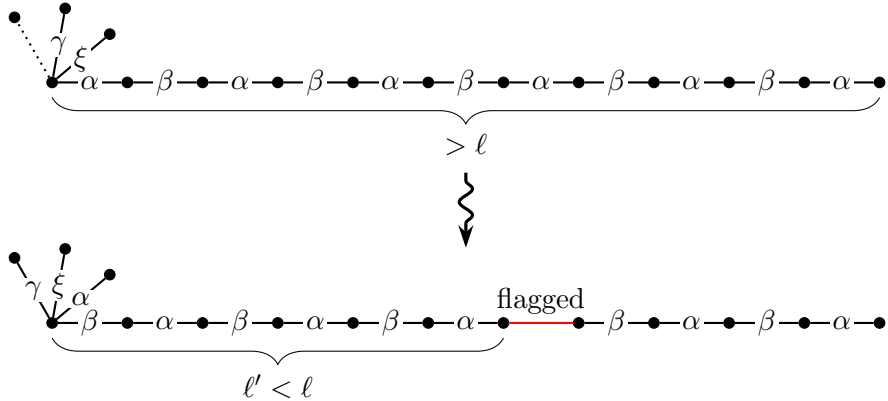
\begin{figure}[htb!]
    \centering
        \begin{tikzpicture}
            \node[circle,fill=black,draw,inner sep=0pt,minimum size=4pt] (a) at (0,0) {};
        	\path (a) ++(0:1) node[circle,fill=black,draw,inner sep=0pt,minimum size=4pt] (b) {};
        	\path (b) ++(0:1) node[circle,fill=black,draw,inner sep=0pt,minimum size=4pt] (c) {};
                \path (c) ++(0:1) node[circle,fill=black,draw,inner sep=0pt,minimum size=4pt] (d) {};
                \path (d) ++(0:1) node[circle,fill=black,draw,inner sep=0pt,minimum size=4pt] (e) {};
                \path (e) ++(0:1) node[circle,fill=black,draw,inner sep=0pt,minimum size=4pt] (l) {};
                \path (l) ++(0:1) node[circle,fill=black,draw,inner sep=0pt,minimum size=4pt] (m) {};
                \path (m) ++(0:1) node[circle,fill=black,draw,inner sep=0pt,minimum size=4pt] (n) {};
                \path (n) ++(0:1) node[circle,fill=black,draw,inner sep=0pt,minimum size=4pt] (o) {};
                \path (o) ++(0:1) node[circle,fill=black,draw,inner sep=0pt,minimum size=4pt] (f) {};
                \path (f) ++(0:1) node[circle,fill=black,draw,inner sep=0pt,minimum size=4pt] (g) {};
                \path (g) ++(0:1) node[circle,fill=black,draw,inner sep=0pt,minimum size=4pt] (h) {};

                \path (a) ++(40:1) node[circle,fill=black,draw,inner sep=0pt,minimum size=4pt] (i) {};
                \path (a) ++(80:1) node[circle,fill=black,draw,inner sep=0pt,minimum size=4pt] (j) {};
                \path (a) ++(120:1) node[circle,fill=black,draw,inner sep=0pt,minimum size=4pt] (k) {};

                \draw[thick] (a) to node[midway,inner sep=1pt,outer sep=1pt,minimum size=4pt,fill=white] {$\alpha$} (i) (a) to node[midway,inner sep=1pt,outer sep=1pt,minimum size=4pt,fill=white] {$\xi$} (j) (a) to node[midway,inner sep=1pt,outer sep=1pt,minimum size=4pt,fill=white] {$\gamma$} (k);
         
        	\draw[thick] (a) to node[midway,inner sep=1pt,outer sep=1pt,minimum size=4pt,fill=white] {$\beta$} (b) to node[midway,inner sep=1pt,outer sep=1pt,minimum size=4pt,fill=white] {$\alpha$} (c) to node[midway,inner sep=1pt,outer sep=1pt,minimum size=4pt,fill=white] {$\beta$} (d) to node[midway,inner sep=1pt,outer sep=1pt,minimum size=4pt,fill=white] {$\alpha$} (e) to node[midway,inner sep=1pt,outer sep=1pt,minimum size=4pt,fill=white] {$\beta$} (l) to node[midway,inner sep=1pt,outer sep=1pt,minimum size=4pt,fill=white] {$\alpha$} (m)  (n) to node[midway,inner sep=1pt,outer sep=1pt,minimum size=4pt,fill=white] {$\beta$} (o) to node[midway,inner sep=1pt,outer sep=1pt,minimum size=4pt,fill=white] {$\alpha$} (f) to node[midway,inner sep=1pt,outer sep=1pt,minimum size=4pt,fill=white] {$\beta$} (g) to node[midway,inner sep=1pt,outer sep=1pt,minimum size=4pt,fill=white] {$\alpha$} (h);

                \draw[thick, red] (m) -- (n);
                \node at (6.5, 0.3) {flagged};

            \draw[decoration={brace,amplitude=10pt,mirror}, decorate] (0, -0.2) -- node [midway,below,xshift=0pt,yshift=-10pt] {$\ell'< \ell$} (6,-0.2);

        \begin{scope}[yshift=3.2cm]
            \node[circle,fill=black,draw,inner sep=0pt,minimum size=4pt] (a) at (0,0) {};
        	\path (a) ++(0:1) node[circle,fill=black,draw,inner sep=0pt,minimum size=4pt] (b) {};
        	\path (b) ++(0:1) node[circle,fill=black,draw,inner sep=0pt,minimum size=4pt] (c) {};
                \path (c) ++(0:1) node[circle,fill=black,draw,inner sep=0pt,minimum size=4pt] (d) {};
                \path (d) ++(0:1) node[circle,fill=black,draw,inner sep=0pt,minimum size=4pt] (e) {};
                \path (e) ++(0:1) node[circle,fill=black,draw,inner sep=0pt,minimum size=4pt] (l) {};
                \path (l) ++(0:1) node[circle,fill=black,draw,inner sep=0pt,minimum size=4pt] (m) {};
                \path (m) ++(0:1) node[circle,fill=black,draw,inner sep=0pt,minimum size=4pt] (n) {};
                \path (n) ++(0:1) node[circle,fill=black,draw,inner sep=0pt,minimum size=4pt] (o) {};
                \path (o) ++(0:1) node[circle,fill=black,draw,inner sep=0pt,minimum size=4pt] (f) {};
                \path (f) ++(0:1) node[circle,fill=black,draw,inner sep=0pt,minimum size=4pt] (g) {};
                \path (g) ++(0:1) node[circle,fill=black,draw,inner sep=0pt,minimum size=4pt] (h) {};

                \path (a) ++(40:1) node[circle,fill=black,draw,inner sep=0pt,minimum size=4pt] (i) {};
                \path (a) ++(80:1) node[circle,fill=black,draw,inner sep=0pt,minimum size=4pt] (j) {};
                \path (a) ++(120:1) node[circle,fill=black,draw,inner sep=0pt,minimum size=4pt] (k) {};

                \draw[thick] (a) to node[midway,inner sep=1pt,outer sep=1pt,minimum size=4pt,fill=white] {$\xi$} (i) (a) to node[midway,inner sep=1pt,outer sep=1pt,minimum size=4pt,fill=white] {$\gamma$} (j);

                \draw[thick, dotted] (a) -- (k);
         
        	\draw[thick] (a) to node[midway,inner sep=1pt,outer sep=1pt,minimum size=4pt,fill=white] {$\alpha$} (b) to node[midway,inner sep=1pt,outer sep=1pt,minimum size=4pt,fill=white] {$\beta$} (c) to node[midway,inner sep=1pt,outer sep=1pt,minimum size=4pt,fill=white] {$\alpha$} (d) to node[midway,inner sep=1pt,outer sep=1pt,minimum size=4pt,fill=white] {$\beta$} (e) to node[midway,inner sep=1pt,outer sep=1pt,minimum size=4pt,fill=white] {$\alpha$} (l) to node[midway,inner sep=1pt,outer sep=1pt,minimum size=4pt,fill=white] {$\beta$} (m) to node[midway,inner sep=1pt,outer sep=1pt,minimum size=4pt,fill=white] {$\alpha$} (n) to node[midway,inner sep=1pt,outer sep=1pt,minimum size=4pt,fill=white] {$\beta$} (o) to node[midway,inner sep=1pt,outer sep=1pt,minimum size=4pt,fill=white] {$\alpha$} (f) to node[midway,inner sep=1pt,outer sep=1pt,minimum size=4pt,fill=white] {$\beta$} (g) to node[midway,inner sep=1pt,outer sep=1pt,minimum size=4pt,fill=white] {$\alpha$} (h);

            \draw[decoration={brace,amplitude=10pt,mirror}, decorate] (0, -0.2) -- node [midway,below,xshift=0pt,yshift=-10pt] {$> \ell$} (11,-0.2);
            
        \end{scope}

        \begin{scope}[yshift=2.5cm]
            \draw[-{Stealth[length=3mm,width=2mm]},very thick,decoration = {snake,pre length=3pt,post length=7pt,},decorate] (5.5,-0.5) -- (5.5,-1.5);
        \end{scope}
        	
        \end{tikzpicture}
    \caption{The process of augmenting an initial segment of a Vizing chain after flagging an edge.}
    \label{fig:main_change}
\end{figure}

For Step~\ref{step:random_edge}, Duan, He, and Zhang design a rather complex randomized procedure to select an uncolored edge at Step~\ref{step:choose_S}.
First, they partition $V$ into $k \defeq \Theta(\log \Delta)$ sets $V_1, \ldots, V_k$ such that $v \in V_i$ if $\deg_G(v) \in [2^{i-1}, 2^i)$.
Starting with $j = k$, they sample a vertex $v \in V_j$ uniformly at random and an uncolored edge incident to $v$ to color.
They continue in this fashion until all edges with an endpoint in $V_j$ are colored.
Once this condition is satisfied, they update $j$ to $j-1$.
As mentioned earlier, a key part of their argument is to show that their algorithm constructs a Vizing chain $(F, P)$ where $P$ is short within $\poly(\log n, 1/\eps)$ attempts.
The analysis is rather lengthy and technical with this sampling procedure playing a crucial role.
As a result of our flagging procedure described earlier, we may simply sample an edge uniformly at random among all uncolored edges.
This can, of course, be implemented very simply and we are therefore able to avoid the complex arguments of \cite{duan2019dynamic}.

Note that we need not construct the entire Vizing chain $(F, P)$.
Therefore, we may implement this procedure in $O(\log^2n/\eps^2 + \ell)$ time (see the discussion after Algorithm~\ref{alg:viz}).
An algorithmic sketch of the first stage is provided in Algorithm~\ref{temp:stage1}.

{
\floatname{algorithm}{Algorithm Sketch}
\begin{algorithm}[htb!]\small
    \caption{Stage 1 of our $q$-edge-coloring algorithm}\label{temp:stage1}
    \begin{flushleft}
        \textbf{Input}: A graph $G = (V,E)$ of maximum degree $\Delta$. \\
        \textbf{Output}: A partial $q$-edge-coloring of $G$.
    \end{flushleft}
    \begin{algorithmic}[1]
        \State $\phi \gets$ the empty coloring
        \While{there are uncolored and unflagged edges}
            \State Pick an uncolored and unflagged edge $e$ uniformly at random.
            \State\label{step: choose_C} Find an $e$-augmenting Vizing chain $(F, P)$ using a random sample $C$ of $O(\log n / \eps)$ colors.
            \If{$\length(P) < \ell$}
                \State Augment $\phi$ using $(F, P)$.
            \Else
                \State Pick one of the first $\ell$ edges of $P$ uniformly at random and \textit{flag} it.
                \State Augment the corresponding \textit{initial segment} of the Vizing chain $(F, P)$.
            \EndIf
        \EndWhile
        \State \Return $\phi$
    \end{algorithmic}
\end{algorithm}
}

After Stage 1, we have a partial $q$-edge-coloring of $G$ where the uncolored edges are precisely those that were flagged at some iteration of the \textsf{while loop} in Algorithm~\ref{temp:stage1}.
Let $G'$ be the graph induced by the flagged edges. 
The goal is to color $G'$ using the folklore $(2+\eps)\Delta$-edge-coloring algorithm.
It is now enough to show that $\Delta(G') = O(\eps\Delta)$ with high probability.
This constitutes the heart of our analysis.
We provide a sketch of the proof below.

Consider a vertex $v \in V$ and the $i$-th iteration of the \textsf{while loop} in Algorithm~\ref{temp:stage1}.
In order for $v$ to be incident to the edge flagged during this iteration, it must be the case that $v$ lies on the path $P$ and the flagged edge is incident to $v$.
Note that a maximal alternating path $P$ is uniquely determined by a vertex on $P$ and the colors of the edges on $P$.
Therefore, given $v$ and the palette $C$, there are at most $\binom{|C|}{2} = O(\log^2n/\eps^2)$ possible such paths.
For each such path $P$, there are at most $2\Delta$ choices for the edge $e$ such that $P$ is the path contained in the $e$-augmenting Vizing chain computed (since $e$ must be incident to one of the endpoints of $P$).
Note that during each iteration, the number of uncolored and unflagged edges decreases by precisely $1$ (since the flagged edge is colored at the start of the iteration).
Therefore, as the uncolored edge $e$ and the flagged edge $f$ are chosen uniformly at random, we may conclude that
\[\Pr(v\text{ is incident to the flagged edge during the } i\text{-th iteration}) \,=\, O\left(\frac{1}{(m-i+1)\ell}\,\frac{\Delta \log^2n}{\eps^2}\right),\]
where we use the fact that during the $i$-th iteration, there are $m - i + 1$ uncolored and unflagged edges.
In particular, we have
\[\E[\deg_{G'}(v)] \,=\, \sum_{i = 1}^mO\left(\frac{1}{(m-i+1)\ell}\,\frac{\Delta \log^2n}{\eps^2}\right) \,=\, O\left(\frac{\Delta \log^3n}{\ell\,\eps^2}\right).\]
(Here we use that $1 + 1/2 + 1/3 + \cdots + 1/m \approx \log m$ and $m \leq n^2$.) 
Using a version of Azuma's inequality due to Kuszmaul and Qi \cite{azuma} (see Theorem~\ref{theo:azuma_supermartingale}), we are able to prove a concentration result showing that $\deg(v) = O(\eps\Delta)$ for $\ell = \Theta(\log^3n/\eps^3)$ with high probability; the details are given in \S\ref{sec:proof}.
The bound on $\Delta(G')$ then follows by a union bound over $V$.

By our choice of $\ell$, we may implement Stage 1 in $O(m\log^3n/\eps^3)$ time.
The folklore $(2+\eps)\Delta$-edge-coloring algorithm takes $O(\max\set{m/\eps^2,\,\log^2n/\eps})$ time (see Proposition~\ref{prop:greedy}).
Therefore, the overall runtime of our algorithm is $O(m\log^3n/\eps^3)$, as claimed.

\subsection{Concluding Remarks}\label{subsection: conclusion}

In this work, we consider the algorithmic task of edge-coloring simple graphs with $(1+\eps)\Delta$ colors.
Our focus is on simplicity, both in terms of implementation as well as analysis.
We design an $\tilde O(m/\eps^3)$-time randomized sequential algorithm to solve this task for any graph with $m$ edges having maximum degree $\Delta = \Omega(\log n / \eps)$.
The result improves upon that of \cite{duan2019dynamic}, which was the first near-linear time algorithm for this problem.
We conclude this introduction with potential future directions of inquiry.

The dependence on $\eps$ in the result of Theorem~\ref{theo:main_theo} is rather strong.
In particular, for $\eps = \Theta(\log n/\Delta)$, the runtime of our algorithm is $O(\Delta^3\,m)$, which is very slow for large values of $\Delta$.
The running time of Assadi's algorithm in \cite{assadi2024faster} depends logarithmically on $\eps$ in expectation (see Table~\ref{table:history}), which leads us to ask the following:

\begin{question}
    Is there a simple sequential algorithm for $(1+\eps)\Delta$-edge-coloring with running time $O\left(m\,\poly(\log n,\,\log (1/\eps))\right)$?
\end{question}

Additionally, there are a number of algorithmic results for edge-coloring multigraphs; see, e.g., \cite{dhawan2024edge, gabow1982algorithms, schrijver1998bipartite, alon2003simple}.
In this setting, Observation~\ref{obs: vizing} does not hold for Vizing's theorem.
In particular, it is unclear how to extend the palette sampling idea of \cite{duan2019dynamic} to multigraphs of maximum multiplicity at least $2$.
One could instead consider Shannon's bound on the chromatic index: $\chi'(G) \leq 1.5\Delta$ \cite{Shannon}.

\begin{question}
    Let $G$ be a multigraph of maximum degree $\Delta$ and maximum multiplicity $\mu$.
    For $\Delta \geq (\mu - 1)/\eps$, is there a simple near-linear time sequential algorithm for $(1+\eps)\Delta$-edge-coloring?
    How about for $(1.5+\eps)\Delta$-edge-coloring?
\end{question}

The rest of the paper is structured as follows: in \S\ref{sec:prelim}, we introduce some terminology and background facts that will be used in \S\ref{sec:alg}, where we prove Theorem~\ref{theo:main_theo}.
\section{Notation and Preliminaries}\label{sec:prelim}

This section is split into three subsections. 
In the first, we will introduce some definitions and notation regarding Vizing chains.
In the second, we will describe the data structures we use to store the graph $G$ and attributes regarding a partial coloring $\phi$.
In the third, we will describe the folklore algorithm for $(2+\eps)\Delta$-edge-coloring, which will constitute Stage 2 of our main algorithm.

\subsection{Vizing Chains}\label{subsec:chains}

For $q \in \N$, given a partial $q$-edge-coloring $\phi$ and $x \in V$, we let
\[M(\phi, x) \defeq [q]\setminus\{\phi(xy)\,:\, y \in N_G(x)\}\] 
be the set of all the \emphd{missing} colors at $x$ under the coloring $\phi$.
We note that $|M(\phi, x)| \geq \epsilon\Delta$ for $q = (1+\eps)\Delta$.
An uncolored edge $xy$ is \emphd{$\phi$-happy} if $M(\phi, x)\cap M(\phi, y)\neq \0$. 
If $e = xy$ is $\phi$-happy, we can extend the coloring $\phi$ by assigning any color in $M(\phi, x)\cap M(\phi, y)$ to $e$.

Given a proper partial coloring, we wish to modify it in order to create a partial coloring with a happy edge.
We will do so by constructing so-called called \emphd{Vizing chains} (see Definition~\ref{defn:viz}).
A Vizing chain consists of a \emphd{fan} and an \emphd{alternating path}.
Let us first describe fans.

\begin{definition}[Fans]\label{defn:fans}
    A \emphd{fan} of length $k$ under a partial coloring $\phi$ is a sequence $F = (x, y_0, \ldots, y_{k-1})$ such that:
    \begin{itemize}
        \item $y_0$, \ldots, $y_{k-1}$ are distinct neighbors of $x$,
        \item $\phi(xy_0) = \blank$, and
        \item $\phi(xy_i) \in M(\phi,y_{i-1})$ for $1 \leq i < k$.
    \end{itemize}
    We refer to $x$ as the \emphd{pivot} of the fan, and let $\Pivot(F) \defeq x$, $\vstart(F) \defeq y_0$, and $\vend(F) \defeq y_{k-1}$ denote the pivot, start, and end vertices of a fan $F$. (This notation is uniquely determined unless $k = 1$.)
    Finally, we let $\length(F) = k$.
\end{definition}

The following fact will be useful in our proofs.

\begin{fact}\label{fact:fan}
    Let $\phi$ be a proper partial $q$-edge-coloring and let $F = (x, y_0, \ldots, y_{k-1})$ be a fan under $\phi$.
    Define the coloring $\psi$ as follows:
    \[\psi(e) \defeq \left\{\begin{array}{cc}
        \phi(xy_{i+1}) & \text{if} \quad e = xy_i \quad \text{for} \quad 0 \leq i < k-1; \\
        \blank & \text{if} \quad e = xy_{k-1}; \\
        \phi(e) & \text{otherwise.}
    \end{array}\right.\]
    Then, $\psi$ is a proper partial $q$-edge-coloring.
\end{fact}

The proof of Fact~\ref{fact:fan} is standard and so we omit it here.
We say the coloring $\psi$ is obtained from $\phi$ by \emphd{shifting} the fan $F$.
Such a shifting procedure is described in Fig.~\ref{fig:fan}.

\begin{figure}[htb!]
	\centering
	\begin{tikzpicture}[xscale=1.05]
	\begin{scope}
	\node[circle,fill=black,draw,inner sep=0pt,minimum size=4pt] (x) at (0,0) {};
	\node[anchor=north] at (x) {$x$};
	
	\coordinate (O) at (0,0);
	\def\radius{2.6cm}
	
	\node[circle,fill=black,draw,inner sep=0pt,minimum size=4pt] (y0) at (190:\radius) {};
	\node at (190:2.9) {$y_0$};
	
	\node[circle,fill=black,draw,inner sep=0pt,minimum size=4pt] (y1) at (165:\radius) {};
	\node at (165:2.9) {$y_1$};
	
	\node[circle,fill=black,draw,inner sep=0pt,minimum size=4pt] (y2) at (140:\radius) {};
	\node at (140:2.9) {$y_2$};
	
	\node[circle,fill=black,draw,inner sep=0pt,minimum size=4pt] (y4) at (90:\radius) {};
	\node at (90:2.9) {$y_{i-1}$};
	
	\node[circle,fill=black,draw,inner sep=0pt,minimum size=4pt] (y5) at (65:\radius) {};
	\node at (65:2.9) {$y_i$};
	
	\node[circle,fill=black,draw,inner sep=0pt,minimum size=4pt] (y6) at (40:\radius) {};
	\node at (40:3) {$y_{i+1}$};
	
	\node[circle,fill=black,draw,inner sep=0pt,minimum size=4pt] (y8) at (-10:\radius) {};
	\node at (-10:3.1) {$y_{k-1}$};
	
	\node[circle,inner sep=0pt,minimum size=4pt] at (115:2) {$\ldots$}; 
	\node[circle,inner sep=0pt,minimum size=4pt] at (15:2) {$\ldots$}; 
	
	\draw[thick,dotted] (x) to (y0);
	\draw[thick] (x) to node[midway,inner sep=1pt,outer sep=1pt,minimum size=4pt,fill=white] {$\alpha_0$} (y1);
	\draw[thick] (x) to node[midway,inner sep=1pt,outer sep=1pt,minimum size=4pt,fill=white] {$\alpha_1$} (y2);
	
	\draw[thick] (x) to node[midway,inner sep=1pt,outer sep=1pt,minimum size=4pt,fill=white] {$\alpha_{i-2}$} (y4);
	\draw[thick] (x) to node[pos=0.75,inner sep=1pt,outer sep=1pt,minimum size=4pt,fill=white] {$\alpha_{i-1}$} (y5);
	\draw[thick] (x) to node[midway,inner sep=1pt,outer sep=1pt,minimum size=4pt,fill=white] {$\alpha_i$} (y6);
	
	\draw[thick] (x) to node[midway,inner sep=1pt,outer sep=1pt,minimum size=4pt,fill=white] {$\alpha_{k-2}$} (y8);
	\end{scope}
	
	\draw[-{Stealth[length=1.6mm]},very thick,decoration = {snake,pre length=3pt,post length=7pt,},decorate] (2.9,1) -- (5,1);
	
	\begin{scope}[xshift=8.3cm]
	\node[circle,fill=black,draw,inner sep=0pt,minimum size=4pt] (x) at (0,0) {};
	\node[anchor=north] at (x) {$x$};
	
	\coordinate (O) at (0,0);
	\def\radius{2.6cm}
	
	\node[circle,fill=black,draw,inner sep=0pt,minimum size=4pt] (y0) at (190:\radius) {};
	\node at (190:2.9) {$y_0$};
	
	\node[circle,fill=black,draw,inner sep=0pt,minimum size=4pt] (y1) at (165:\radius) {};
	\node at (165:2.9) {$y_1$};
	
	\node[circle,fill=black,draw,inner sep=0pt,minimum size=4pt] (y2) at (140:\radius) {};
	\node at (140:2.9) {$y_2$};
	
	\node[circle,fill=black,draw,inner sep=0pt,minimum size=4pt] (y4) at (90:\radius) {};
	\node at (90:2.9) {$y_{i-1}$};
	
	\node[circle,fill=black,draw,inner sep=0pt,minimum size=4pt] (y5) at (65:\radius) {};
	\node at (65:2.9) {$y_i$};
	
	\node[circle,fill=black,draw,inner sep=0pt,minimum size=4pt] (y6) at (40:\radius) {};
	\node at (40:3) {$y_{i+1}$};
	
	\node[circle,fill=black,draw,inner sep=0pt,minimum size=4pt] (y8) at (-10:\radius) {};
	\node at (-10:3.1) {$y_{k-1}$};
	
	\node[circle,inner sep=0pt,minimum size=4pt] at (115:2) {$\ldots$}; 
	\node[circle,inner sep=0pt,minimum size=4pt] at (15:2) {$\ldots$}; 
	
	\draw[thick] (x) to node[midway,inner sep=1pt,outer sep=1pt,minimum size=4pt,fill=white] {$\alpha_0$} (y0);
	\draw[thick] (x) to node[midway,inner sep=1pt,outer sep=1pt,minimum size=4pt,fill=white] {$\alpha_1$} (y1);
	\draw[thick] (x) to node[midway,inner sep=1pt,outer sep=1pt,minimum size=4pt,fill=white] {$\alpha_2$} (y2);
	
	\draw[thick] (x) to node[midway,inner sep=1pt,outer sep=1pt,minimum size=4pt,fill=white] {$\alpha_{i-1}$} (y4);
	\draw[thick] (x) to node[pos=0.75,inner sep=1pt,outer sep=1pt,minimum size=4pt,fill=white] {$\alpha_i$} (y5);
	\draw[thick] (x) to node[midway,inner sep=1pt,outer sep=1pt,minimum size=4pt,fill=white] {$\alpha_{i+1}$} (y6);
	
	\draw[thick, dotted] (x) to (y8);
	\end{scope}
	\end{tikzpicture}
	\caption{The process of shifting a fan.}\label{fig:fan}
\end{figure}
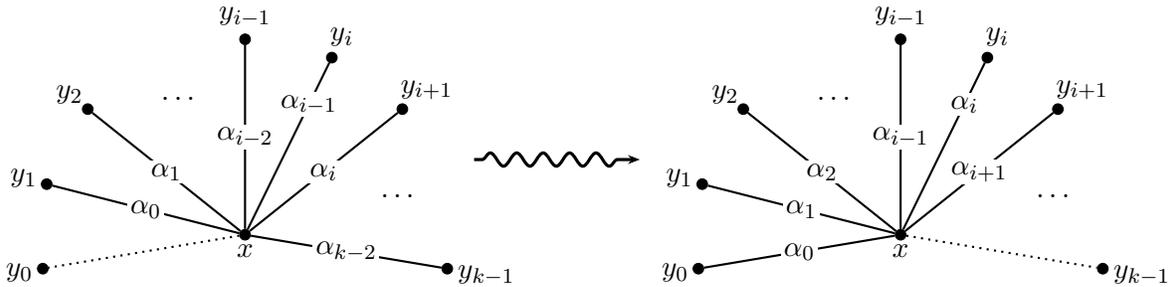

Next, we define alternating paths.

\begin{definition}[Alternating Paths]\label{defn:path}
    For $\alpha, \beta \in [q]$, an \emphd{$\alpha\beta$-path} $P = (x_0, \ldots, x_k)$ under a partial $q$-edge-coloring $\phi$ is a sequence $P = (x_0, \ldots, x_k)$ such that:
    \begin{itemize}
        \item $x_0$, \ldots, $x_k$ are distinct vertices,
        \item $\phi(x_0x_1) = \alpha$, and
        \item the colors of the edges $x_ix_{i+1}$ alternate between $\alpha$ and $\beta$.
    \end{itemize}
    We let $\vstart(P) \defeq x_0$ and $\vend(P) \defeq x_k$ denote the first and last vertices on the path, respectively.
    Finally, we let $\length(P) = k$.
\end{definition}

The following fact will be useful in our proofs.

\begin{fact}\label{fact:path}
    Let $\phi$ be a proper partial $q$-edge-coloring and let $P = (x_0, \ldots, x_k)$ be a maximal $\alpha\beta$-path under $\phi$, i.e., $M(\phi, v) \cap \set{\alpha, \beta} \neq \0$ for $v \in \set{x_0, x_k}$.
    Define the coloring $\psi$ by interchanging the colors of the edges in $P$.
    Then, $\psi$ is a proper partial $q$-edge-coloring.
\end{fact}

The proof of Fact~\ref{fact:path} is standard and so we omit it here.
We say the coloring $\psi$ is obtained from $\phi$ by \emphd{flipping} the path $P$.
Such a flipping procedure is described in Fig.~\ref{fig:path}.

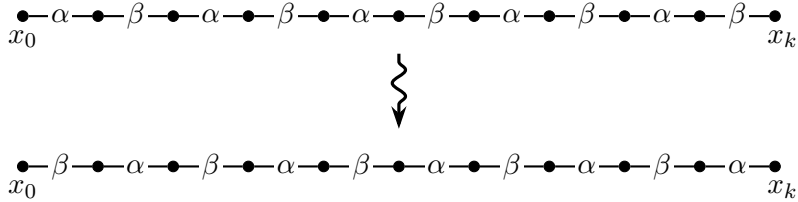
\begin{figure}[htb!]
    \centering
        \begin{tikzpicture}
            \node[circle,fill=black,draw,inner sep=0pt,minimum size=4pt] (a) at (0,0) {};
        	\path (a) ++(0:1) node[circle,fill=black,draw,inner sep=0pt,minimum size=4pt] (b) {};
        	\path (b) ++(0:1) node[circle,fill=black,draw,inner sep=0pt,minimum size=4pt] (c) {};
                \path (c) ++(0:1) node[circle,fill=black,draw,inner sep=0pt,minimum size=4pt] (d) {};
                \path (d) ++(0:1) node[circle,fill=black,draw,inner sep=0pt,minimum size=4pt] (e) {};
                \path (e) ++(0:1) node[circle,fill=black,draw,inner sep=0pt,minimum size=4pt] (f) {};
                \path (f) ++(0:1) node[circle,fill=black,draw,inner sep=0pt,minimum size=4pt] (g) {};
                \path (g) ++(0:1) node[circle,fill=black,draw,inner sep=0pt,minimum size=4pt] (h) {};
                \path (h) ++(0:1) node[circle,fill=black,draw,inner sep=0pt,minimum size=4pt] (i) {};
                \path (i) ++(0:1) node[circle,fill=black,draw,inner sep=0pt,minimum size=4pt] (j) {};
                \path (j) ++(0:1) node[circle,fill=black,draw,inner sep=0pt,minimum size=4pt] (k) {};
        	

            \node at (0, -0.3) {$x_0$};
            \node at (10.1, -0.3) {$x_k$};
        	
        	\draw[thick] (a) to node[midway,inner sep=1pt,outer sep=1pt,minimum size=4pt,fill=white] {$\beta$} (b) to node[midway,inner sep=1pt,outer sep=1pt,minimum size=4pt,fill=white] {$\alpha$} (c) to node[midway,inner sep=1pt,outer sep=1pt,minimum size=4pt,fill=white] {$\beta$} (d) to node[midway,inner sep=1pt,outer sep=1pt,minimum size=4pt,fill=white] {$\alpha$} (e) to node[midway,inner sep=1pt,outer sep=1pt,minimum size=4pt,fill=white] {$\beta$} (f) to node[midway,inner sep=1pt,outer sep=1pt,minimum size=4pt,fill=white] {$\alpha$} (g) to node[midway,inner sep=1pt,outer sep=1pt,minimum size=4pt,fill=white] {$\beta$} (h) to node[midway,inner sep=1pt,outer sep=1pt,minimum size=4pt,fill=white] {$\alpha$} (i) to node[midway,inner sep=1pt,outer sep=1pt,minimum size=4pt,fill=white] {$\beta$} (j) to node[midway,inner sep=1pt,outer sep=1pt,minimum size=4pt,fill=white] {$\alpha$} (k);

        \begin{scope}[yshift=2cm]
            \node[circle,fill=black,draw,inner sep=0pt,minimum size=4pt] (a) at (0,0) {};
        	\path (a) ++(0:1) node[circle,fill=black,draw,inner sep=0pt,minimum size=4pt] (b) {};
        	\path (b) ++(0:1) node[circle,fill=black,draw,inner sep=0pt,minimum size=4pt] (c) {};
                \path (c) ++(0:1) node[circle,fill=black,draw,inner sep=0pt,minimum size=4pt] (d) {};
                \path (d) ++(0:1) node[circle,fill=black,draw,inner sep=0pt,minimum size=4pt] (e) {};
                \path (e) ++(0:1) node[circle,fill=black,draw,inner sep=0pt,minimum size=4pt] (f) {};
                \path (f) ++(0:1) node[circle,fill=black,draw,inner sep=0pt,minimum size=4pt] (g) {};
                \path (g) ++(0:1) node[circle,fill=black,draw,inner sep=0pt,minimum size=4pt] (h) {};
                \path (h) ++(0:1) node[circle,fill=black,draw,inner sep=0pt,minimum size=4pt] (i) {};
                \path (i) ++(0:1) node[circle,fill=black,draw,inner sep=0pt,minimum size=4pt] (j) {};
                \path (j) ++(0:1) node[circle,fill=black,draw,inner sep=0pt,minimum size=4pt] (k) {};
        	

            \node at (0, -0.3) {$x_0$};
            \node at (10.1, -0.3) {$x_k$};
        	
        	\draw[thick] (a) to node[midway,inner sep=1pt,outer sep=1pt,minimum size=4pt,fill=white] {$\alpha$} (b) to node[midway,inner sep=1pt,outer sep=1pt,minimum size=4pt,fill=white] {$\beta$} (c) to node[midway,inner sep=1pt,outer sep=1pt,minimum size=4pt,fill=white] {$\alpha$} (d) to node[midway,inner sep=1pt,outer sep=1pt,minimum size=4pt,fill=white] {$\beta$} (e) to node[midway,inner sep=1pt,outer sep=1pt,minimum size=4pt,fill=white] {$\alpha$} (f) to node[midway,inner sep=1pt,outer sep=1pt,minimum size=4pt,fill=white] {$\beta$} (g) to node[midway,inner sep=1pt,outer sep=1pt,minimum size=4pt,fill=white] {$\alpha$} (h) to node[midway,inner sep=1pt,outer sep=1pt,minimum size=4pt,fill=white] {$\beta$} (i) to node[midway,inner sep=1pt,outer sep=1pt,minimum size=4pt,fill=white] {$\alpha$} (j) to node[midway,inner sep=1pt,outer sep=1pt,minimum size=4pt,fill=white] {$\beta$} (k);
        \end{scope}

        \begin{scope}[yshift=1.5cm]
            \draw[-{Stealth[length=3mm,width=2mm]},very thick,decoration = {snake,pre length=3pt,post length=7pt,},decorate] (5,0) -- (5,-1);
        \end{scope}
        	
        \end{tikzpicture}
    \caption{The process of flipping an $\alpha\beta$-path.}
    \label{fig:path}
\end{figure}

Combining a fan and an alternating path yields a Vizing chain.

\begin{definition}[Vizing chains]\label{defn:viz}
    A \emphd{Vizing chain} in a partial $q$-edge-coloring $\phi$ is a tuple $(F, P)$ such that $F = (x, y_0, \ldots, y_{k-1})$ is a fan and $P = (x_0 = x, x_1, \ldots, x_s)$ is an $\alpha\beta$-path for some $\alpha, \beta \in [q]$.
\end{definition}

A key part of our algorithm will involve constructing ``short'' Vizing chains $(F, P)$ and modifying our coloring by flipping $P$ and shifting $F$ to create a happy edge.
Such a procedure is described in Fig.~\ref{fig:vizing_shift}.

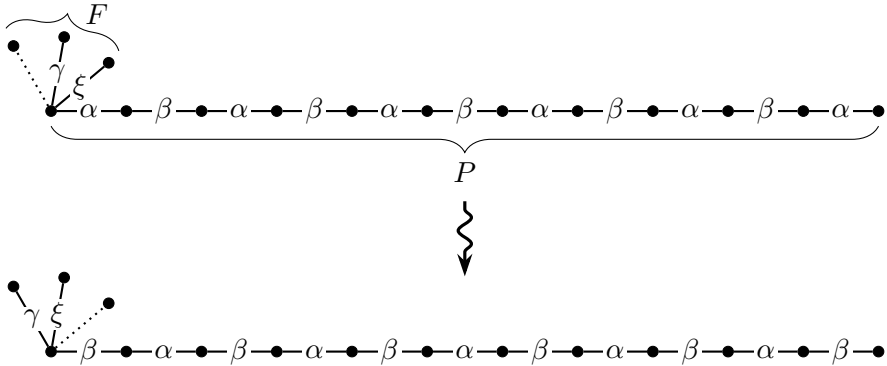
\begin{figure}[htb!]
    \centering
        \begin{tikzpicture}
            \node[circle,fill=black,draw,inner sep=0pt,minimum size=4pt] (a) at (0,0) {};
        	\path (a) ++(0:1) node[circle,fill=black,draw,inner sep=0pt,minimum size=4pt] (b) {};
        	\path (b) ++(0:1) node[circle,fill=black,draw,inner sep=0pt,minimum size=4pt] (c) {};
                \path (c) ++(0:1) node[circle,fill=black,draw,inner sep=0pt,minimum size=4pt] (d) {};
                \path (d) ++(0:1) node[circle,fill=black,draw,inner sep=0pt,minimum size=4pt] (e) {};
                \path (e) ++(0:1) node[circle,fill=black,draw,inner sep=0pt,minimum size=4pt] (l) {};
                \path (l) ++(0:1) node[circle,fill=black,draw,inner sep=0pt,minimum size=4pt] (m) {};
                \path (m) ++(0:1) node[circle,fill=black,draw,inner sep=0pt,minimum size=4pt] (n) {};
                \path (n) ++(0:1) node[circle,fill=black,draw,inner sep=0pt,minimum size=4pt] (o) {};
                \path (o) ++(0:1) node[circle,fill=black,draw,inner sep=0pt,minimum size=4pt] (f) {};
                \path (f) ++(0:1) node[circle,fill=black,draw,inner sep=0pt,minimum size=4pt] (g) {};
                \path (g) ++(0:1) node[circle,fill=black,draw,inner sep=0pt,minimum size=4pt] (h) {};

                \path (a) ++(40:1) node[circle,fill=black,draw,inner sep=0pt,minimum size=4pt] (i) {};
                \path (a) ++(80:1) node[circle,fill=black,draw,inner sep=0pt,minimum size=4pt] (j) {};
                \path (a) ++(120:1) node[circle,fill=black,draw,inner sep=0pt,minimum size=4pt] (k) {};
        	
                \draw[thick, dotted] (a) -- (i);
                
                \draw[thick] (a) to node[midway,inner sep=1pt,outer sep=1pt,minimum size=4pt,fill=white] {$\xi$} (j) (a) to node[midway,inner sep=1pt,outer sep=1pt,minimum size=4pt,fill=white] {$\gamma$} (k);
         
        	\draw[thick] (a) to node[midway,inner sep=1pt,outer sep=1pt,minimum size=4pt,fill=white] {$\beta$} (b) to node[midway,inner sep=1pt,outer sep=1pt,minimum size=4pt,fill=white] {$\alpha$} (c) to node[midway,inner sep=1pt,outer sep=1pt,minimum size=4pt,fill=white] {$\beta$} (d) to node[midway,inner sep=1pt,outer sep=1pt,minimum size=4pt,fill=white] {$\alpha$} (e) to node[midway,inner sep=1pt,outer sep=1pt,minimum size=4pt,fill=white] {$\beta$} (l) to node[midway,inner sep=1pt,outer sep=1pt,minimum size=4pt,fill=white] {$\alpha$} (m) to node[midway,inner sep=1pt,outer sep=1pt,minimum size=4pt,fill=white] {$\beta$} (n) to node[midway,inner sep=1pt,outer sep=1pt,minimum size=4pt,fill=white] {$\alpha$} (o) to node[midway,inner sep=1pt,outer sep=1pt,minimum size=4pt,fill=white] {$\beta$} (f) to node[midway,inner sep=1pt,outer sep=1pt,minimum size=4pt,fill=white] {$\alpha$} (g) to node[midway,inner sep=1pt,outer sep=1pt,minimum size=4pt,fill=white] {$\beta$} (h);

        \begin{scope}[yshift=3.2cm]
            \node[circle,fill=black,draw,inner sep=0pt,minimum size=4pt] (a) at (0,0) {};
        	\path (a) ++(0:1) node[circle,fill=black,draw,inner sep=0pt,minimum size=4pt] (b) {};
        	\path (b) ++(0:1) node[circle,fill=black,draw,inner sep=0pt,minimum size=4pt] (c) {};
                \path (c) ++(0:1) node[circle,fill=black,draw,inner sep=0pt,minimum size=4pt] (d) {};
                \path (d) ++(0:1) node[circle,fill=black,draw,inner sep=0pt,minimum size=4pt] (e) {};
                \path (e) ++(0:1) node[circle,fill=black,draw,inner sep=0pt,minimum size=4pt] (l) {};
                \path (l) ++(0:1) node[circle,fill=black,draw,inner sep=0pt,minimum size=4pt] (m) {};
                \path (m) ++(0:1) node[circle,fill=black,draw,inner sep=0pt,minimum size=4pt] (n) {};
                \path (n) ++(0:1) node[circle,fill=black,draw,inner sep=0pt,minimum size=4pt] (o) {};
                \path (o) ++(0:1) node[circle,fill=black,draw,inner sep=0pt,minimum size=4pt] (f) {};
                \path (f) ++(0:1) node[circle,fill=black,draw,inner sep=0pt,minimum size=4pt] (g) {};
                \path (g) ++(0:1) node[circle,fill=black,draw,inner sep=0pt,minimum size=4pt] (h) {};

                \path (a) ++(40:1) node[circle,fill=black,draw,inner sep=0pt,minimum size=4pt] (i) {};
                \path (a) ++(80:1) node[circle,fill=black,draw,inner sep=0pt,minimum size=4pt] (j) {};
                \path (a) ++(120:1) node[circle,fill=black,draw,inner sep=0pt,minimum size=4pt] (k) {};

                \draw[thick] (a) to node[midway,inner sep=1pt,outer sep=1pt,minimum size=4pt,fill=white] {$\xi$} (i) (a) to node[midway,inner sep=1pt,outer sep=1pt,minimum size=4pt,fill=white] {$\gamma$} (j);

                \draw[thick, dotted] (a) -- (k);
         
        	\draw[thick] (a) to node[midway,inner sep=1pt,outer sep=1pt,minimum size=4pt,fill=white] {$\alpha$} (b) to node[midway,inner sep=1pt,outer sep=1pt,minimum size=4pt,fill=white] {$\beta$} (c) to node[midway,inner sep=1pt,outer sep=1pt,minimum size=4pt,fill=white] {$\alpha$} (d) to node[midway,inner sep=1pt,outer sep=1pt,minimum size=4pt,fill=white] {$\beta$} (e) to node[midway,inner sep=1pt,outer sep=1pt,minimum size=4pt,fill=white] {$\alpha$} (l) to node[midway,inner sep=1pt,outer sep=1pt,minimum size=4pt,fill=white] {$\beta$} (m) to node[midway,inner sep=1pt,outer sep=1pt,minimum size=4pt,fill=white] {$\alpha$} (n) to node[midway,inner sep=1pt,outer sep=1pt,minimum size=4pt,fill=white] {$\beta$} (o) to node[midway,inner sep=1pt,outer sep=1pt,minimum size=4pt,fill=white] {$\alpha$} (f) to node[midway,inner sep=1pt,outer sep=1pt,minimum size=4pt,fill=white] {$\beta$} (g) to node[midway,inner sep=1pt,outer sep=1pt,minimum size=4pt,fill=white] {$\alpha$} (h);

            \draw[decoration={brace,amplitude=10pt,mirror}, decorate] (0, -0.2) -- node [midway,below,xshift=0pt,yshift=-10pt] {$P$} (11,-0.2);
            \draw[decoration={brace,amplitude=10pt},decorate] (-0.6,1.1) -- node [midway,above,yshift=2pt,xshift=13pt] {$F$} (0.9, 0.8);
            
        \end{scope}

        \begin{scope}[yshift=2.5cm]
            \draw[-{Stealth[length=3mm,width=2mm]},very thick,decoration = {snake,pre length=3pt,post length=7pt,},decorate] (5.5,-0.5) -- (5.5,-1.5);
        \end{scope}
        	
        \end{tikzpicture}
    \caption{The process of flipping $P$ and then shifting $F$ in a Vizing chain $(F, P)$.}
    \label{fig:vizing_shift}
\end{figure}

\subsection{Data Structures}\label{subsec:data_structures}
    
In this section, we describe how we will store our graph $G$ and partial coloring $\phi$.
Below, we discuss how our choices for the data structures affect the runtime of certain procedures.

\begin{itemize}
    \item We store $G$ as a list of vertices and edges, and include the partial coloring $\phi$ as an attribute of the graph.

    \item We store the partial coloring $\phi$ as a hash map, which maps edges to their respective colors (or $\blank$ if the edge is uncolored).
    Furthermore, the missing sets $M(\phi, \cdot)$ are also stored as hash maps, which map a vertex $x$ to a $q$-element array such that the following holds for each $\alpha \in [q]$:
    \[M(\phi, x)[\alpha] = \left\{\begin{array}{cc}
        y & \text{such that $\phi(xy) = \alpha$;} \\
        \blank & \text{if no such $y$ exists.}
    \end{array}\right.\]
    Note that as $\phi$ is a proper partial coloring, the vertex $y$ above is unique.
    In the remainder of the paper, we will use the notation $M(\phi, x)[\cdot]$ as described above in our algorithms, and the notation $M(\phi, x)$ to indicate the set of missing colors at $x$ in our proofs.
\end{itemize}
By our choice of data structures, we may shift a fan $F$ (resp. flip a path $P$) in time $O(\length(F))$ (resp. $O(\length(P))$).

\subsection{A Folklore $(2+\eps)\Delta$-Edge-Coloring Algorithm}\label{subsec:stage two}

In this section, we describe the folklore algorithm for $(2+\eps)\Delta$-edge-coloring, which we will invoke during Stage 2 of our main algorithm.
The algorithm proceeds as follows: iterate over the edges in an arbitrary order and at each iteration, repeatedly pick a color from $[(2+\eps)\Delta]$ uniformly at random until a valid one is selected (a color $\alpha$ is valid for an edge $e$ if no edge sharing an endpoint with $e$ is colored $\alpha$).

\vspace{5pt}
\begin{breakablealgorithm}\small
\caption{$(2+\eps)\Delta$-Edge-Coloring}\label{alg:greedy}
\begin{flushleft}
\textbf{Input}: An $n$-vertex graph $G$ having $m$ edges and maximum degree $\Delta$. \\
\textbf{Output}: A proper $(2 + \epsilon)\Delta$-edge-coloring of $G$.
\end{flushleft}

\begin{algorithmic}[1]
    \State $\phi(e) \gets \blank$ \textbf{for each} $e \in E$
    \For{$e = xy \in E$}
        \State \label{step:random color choice} Pick $\alpha \in [(2+\eps)\Delta]$ uniformly at random.
        \If{$M(\phi, x)[\alpha] = \blank$ and $M(\phi, y)[\alpha] = \blank$}
            \State $\phi(e) \gets \alpha$
        \Else
            \State Return to Step~\ref{step:random color choice}.
        \EndIf
    \EndFor
    \State \Return $\phi$
\end{algorithmic}
\end{breakablealgorithm}
\vspace{5pt}

The following proposition provides a bound on the runtime of the above algorithm.
The proof follows a standard coupling argument.
For completeness, we include such a proof in \S\ref{appendix}.

\begin{proposition}\label{prop:greedy}
    Let $\eps > 0$ be arbitrary and let $\gamma \defeq \min\set{1,\, \eps}$.
    For any $n$-vertex graph $G$ with $m$ edges and maximum degree $\Delta$,
    Algorithm~\ref{alg:greedy} computes a proper $(2+\eps)\Delta$-edge-coloring of $G$ in $O(\max\set{m/\gamma^2,\,\log^2n/\gamma})$ time with probability at least $1 - \min\set{e^{-m/\gamma}, 1/\poly(n)}$.
\end{proposition}

\section{Proof of Theorem~\ref{theo:main_theo}}\label{sec:alg}

In this section, we will prove Theorem~\ref{theo:main_theo}.
We will split this section into three subsections.
In the first, we will describe our algorithm formally;
in the second subsection, we will prove the correctness of our algorithm;
and in the third, we will analyze the probability of failure.

\subsection{Algorithm Overview}\label{subsec:overview}

In this section, we will provide an overview of our algorithm (formally stated as Algorithm~\ref{alg:color}).
As mentioned in \S\ref{subsection: overview}, the algorithm will proceed in two stages.
We will split the algorithm for Stage 1 into subprocedures.

The first subprocedure we will describe is \hyperref[alg:fan]{Make Fan}, which is stated formally as Algorithm~\ref{alg:fan}. 
The algorithm takes as input a proper partial coloring $\phi$, an uncolored edge $e = xy$, a choice of a pivot vertex $x \in e$, and a subset of colors $C \subseteq [q]$.
Either the algorithm fails, or it outputs a tuple $(F, \alpha, j)$ such that:
\begin{itemize}
    \item $F = (x, y_0 = y, \ldots, y_{k-1})$ is a fan under the partial coloring $\phi$,
    \item $\alpha \in C$ is a color, and $j \leq k$ is an index such that $\alpha \in M(\phi, y_{k-1})\cap M(\phi, y_{j-1})$.
\end{itemize}
To construct $F$, we follow a series of iterations. At the start of each iteration, we have a fan $F = (x, y_0 = y, \ldots, y_s)$.
If $M(\phi, y_s) \cap C = \0$, the algorithm fails.
If not, we let $\eta \defeq \min M(\phi, y_s) \cap C$.
If $\eta \in M(\phi, x)$, then we return $(F, \eta, s+1)$.
If not, let $z$ be such that $\phi(xz) = \eta$. 
We now have two cases.
\begin{enumerate}[label=\ep{\textbf{Case \arabic*}},wide]
    \item $z\notin \set{y_0, \ldots, y_s}$. Then we update $F$ to $(x, y_0, \ldots, y_s, z)$ and continue.
    \item $z = y_j$ for some $0 \leq j \leq s$. 
    Note that $\phi(xy_0) = \blank$ and $\eta \in M(\phi, y_s)$, so we must have $1 \leq j \leq s - 1$.
    In this case, we return $(F, \eta, j)$.
\end{enumerate}

\begin{breakablealgorithm}\small
\caption{Make Fan}\label{alg:fan}
\begin{flushleft}
\textbf{Input}: A proper partial $q$-edge-coloring $\phi$, an uncolored edge $e = xy$, a vertex $x \in e$, and list of colors $C \subseteq [q]$. \\
\textbf{Output}: $\mathsf{FAIL}$ or a fan $F = (x, y_0 = y, \ldots, y_{k-1})$, a color $\alpha \in C$, and an index $j$ such that $\alpha \in M(\phi, y_{k-1})\cap M(\phi, y_{j-1})$.
\end{flushleft}

\begin{algorithmic}[1]
    \State $F \gets (x,y)$, \quad $z \gets y$, \quad $k \gets 1$.
    \While{$\mathsf{true}$}
        \If{$C \cap M(\phi, z) = \0$} \label{step:check_missing}
            \State \Return $\mathsf{FAIL}$ \label{step:fail_fan}
        \EndIf
        \medskip
        \State $\eta \gets \min M(\phi, z) \cap C$ \label{step:assign_eta}
        \State $z \gets M(\phi, x)[\eta]$
        \If{$z = \blank$}
            \State \Return $(F, \eta, k)$ \label{step:happy_fan}
        \EndIf
        \medskip
        \If{$z \in F$}\label{step:check_fan}
            \State Let $j$ be such that $z = y_j$.
            \State \Return $(F, \eta, j)$
        \EndIf
        \medskip
        \State $y_k \gets z$
        \State $\mathsf{Append}(F, y_k)$, \quad $k \gets k + 1$
    \EndWhile \label{step:end}
\end{algorithmic}
\end{breakablealgorithm}
\vspace{5pt}

Note that Steps~\ref{step:check_missing} and \ref{step:assign_eta} can be implemented in $O(|C|)$ time and Step~\ref{step:check_fan} takes $O(k)$ time.
It is not hard to see that the number of iterations is at most $|C|$ and so $k \leq |C| + 1$.
It follows that Algorithm~\ref{alg:fan} runs in $O(|C|^2)$ time.

The next algorithm we will describe is the \hyperref[alg:viz]{Vizing Chain Algorithm}, which is stated formally as Algorithm~\ref{alg:viz}. 
The algorithm takes as input a partial coloring $\phi$, an uncolored edge $e = xy$, a choice of a pivot vertex $x \in e$, a subset of colors $C \subseteq [q]$, and a parameter $\ell$ to be defined later.
Either the algorithm fails, or it outputs a Vizing chain $(F, P)$ and a color $\alpha$, which satisfy certain properties we will describe and prove in the next subsection.
First, we run Algorithm~\ref{alg:fan} with input $(\phi, e, x, C)$.
If the algorithm fails, we return $\mathsf{FAIL}$.
If not, let $(F, \alpha, j)$ be the output of Algorithm~\ref{alg:fan} such that $F = (x, y_0, \ldots, y_{k-1})$.
If $j = k$, we return the Vizing chain $(F,(x))$ and the color $\alpha$.
If $M(\phi, x) \cap C = \0$, we return $\mathsf{FAIL}$.
If not, let $\beta \defeq \min M(\phi, x) \cap C$, and let $P = (x_0 = x, x_1, \ldots, x_s)$ be the $\alpha\beta$-path starting at $x$ such that $s \leq \ell$ (i.e., if the maximal $\alpha\beta$-path starting at $x$ has length $> \ell$, we only consider the first $\ell$ edges).
We return the Vizing chain $(F,P)$ and the color $\alpha$.

\vspace{5pt}
\begin{breakablealgorithm}\small
\caption{Vizing Chain Algorithm}\label{alg:viz}
\begin{flushleft}
\textbf{Input}: A proper partial coloring $\phi$, an uncolored edge $e = xy$, a choice of a pivot vertex $x \in e$, a subset of colors $C \subseteq [(1+\epsilon)\Delta]$, and a parameter $\ell$. \\
\textbf{Output}: $\mathsf{FAIL}$ or a Vizing chain $(F,P)$ and a color $\alpha$.
\end{flushleft}

\begin{algorithmic}[1]
    \State $\mathsf{Out} \gets \hyperref[alg:fan]{\mathsf{MakeFan}}(\phi, xy, x, C)$ \Comment{Algorithm~\ref{alg:fan}} \label{step:alg_fan}
    \If{$\mathsf{Out} = \mathsf{FAIL}$}
        \State \Return $\mathsf{FAIL}$
    \EndIf
    \medskip
    \State $(F, \alpha, j) \gets \mathsf{Out}$
    \If{$j = \length(F)$} 
        \State \Return $(F, (x)), \alpha$ \label{step:happy_fan_viz}
    \EndIf
    \medskip
    \If{$M(\phi, x) \cap C  = \0$}
        \State \Return $\mathsf{FAIL}$ \label{step:fail_x_alpha}
    \EndIf
    \medskip
    \State $\beta \gets \min M(\phi, x) \cap C$ \label{step:alpha}
    \State Let $P$ be the $\alpha\beta$-path of length at most $\ell$ starting at $x$ \label{step:path_def}
    \State \Return $(F, P), \alpha$.
\end{algorithmic}
\end{breakablealgorithm}
\vspace{5pt}

Note that we only compute a path of length at most $\ell$ at Step~\ref{step:path_def}.
Due to the structure of $M(\phi, \cdot)$, we conclude that we may implement Algorithm~\ref{alg:viz} in $O(|C|^2 + \ell)$ time.

We are now ready to describe our algorithm.
The formal description is in Algorithm~\ref{alg:color}, but we first provide an informal overview.
The algorithm takes as input a graph $G$ having maximum degree $\Delta$ and two parameters $\kappa$ and $\ell$, and outputs a proper $(1+\eps)\Delta$-edge-coloring of $G$.
As mentioned earlier, the algorithm proceeds in two stages.

During the first stage, we will construct a partial $(1 + \epsilon/2)\Delta$-edge-coloring of $G$.
We compute this coloring iteratively.
We begin by considering all edges to be unflagged as well as uncolored.
Let $\phi_i$ be the coloring at the start of the $i$-th iteration.
We will first select an unflagged and uncolored edge $e$ and a vertex $x \in e$ uniformly at random.
Next, we define $C \subseteq [(1 + \epsilon/2)\Delta]$ by sampling with replacement $\kappa$ times from $[(1 + \epsilon/2)\Delta]$.
We will run Algorithm~\ref{alg:viz} with input $(\phi_i, e, x, C, \ell)$.
If the output is $\mathsf{FAIL}$, we return $\mathsf{FAIL}$.
If not, let $((F, P), \alpha)$ be the output such that $F = (x, y_0, \ldots, y_{k-1})$ and $P = (x_0 = x, x_1, \ldots, x_s)$.
If $s = 0$, then define $\phi_{i+1}$ by shifting $\phi_i$ with $F$ and coloring $xy_{k-1}$ with $\alpha$, and continue on to the next iteration.
If $s = \ell$, let $\ell' \in [\ell]$ be chosen uniformly at random. 
Uncolor and flag the edge $x_{\ell' - 1}x_{\ell'}$ and let $P' = (x_0 = x, \ldots, x_{\ell' - 1})$ (see Fig.~\ref{subfig:flag}).
If $s < \ell$, let $P' = P$.
Note that by construction we have $x_1 = y_j$ for some $0 < j < k-1$.
Let $F' = (x, y_0, \ldots, y_{j-1})$.
We now consider separate cases:
\begin{enumerate}[label=(Case\arabic*), wide]
    \item $\vend(P') \neq y_{j-1}$ (see Fig.~\ref{subfig:case1}).
    Flip the path $P'$, then shift the fan $F'$, and finally color the edge $xy_{j-1}$ with $\alpha$.

    \item $\vend(P') = y_{j-1}$ (see Fig.~\ref{subfig:case2}).
    Flip the path $P'$, then shift the fan $F$, and finally color the edge $xy_{k-1}$ with $\alpha$.
\end{enumerate}

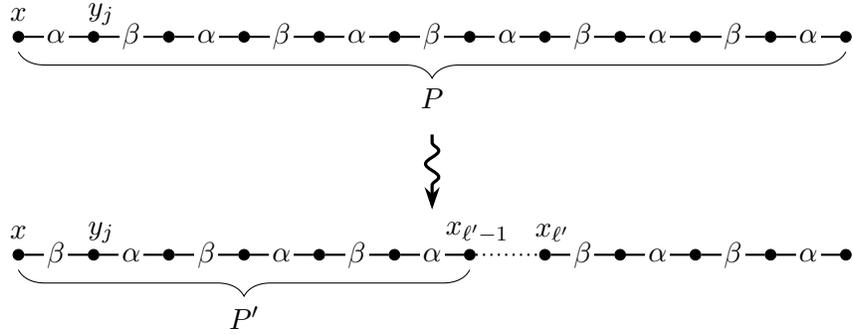
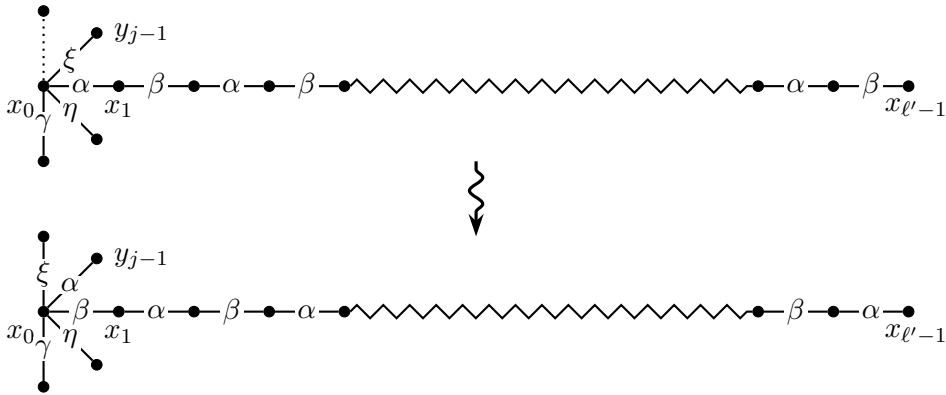
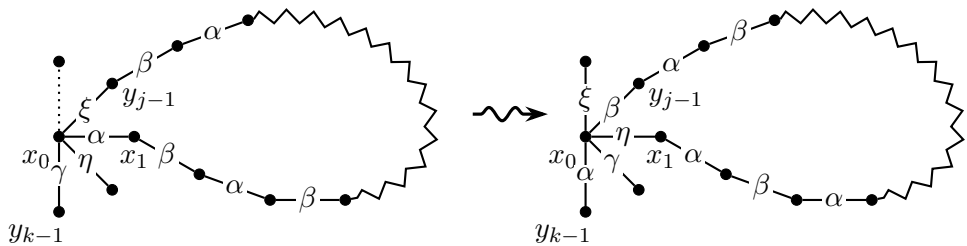
\begin{figure}[b!]
	\centering
	\begin{subfigure}[t]{\textwidth}
		\centering
		\begin{tikzpicture}
            \node[circle,fill=black,draw,inner sep=0pt,minimum size=4pt] (a) at (0,0) {};
        	\path (a) ++(0:1) node[circle,fill=black,draw,inner sep=0pt,minimum size=4pt] (b) {};
        	\path (b) ++(0:1) node[circle,fill=black,draw,inner sep=0pt,minimum size=4pt] (c) {};
                \path (c) ++(0:1) node[circle,fill=black,draw,inner sep=0pt,minimum size=4pt] (d) {};
                \path (d) ++(0:1) node[circle,fill=black,draw,inner sep=0pt,minimum size=4pt] (e) {};
                \path (e) ++(0:1) node[circle,fill=black,draw,inner sep=0pt,minimum size=4pt] (l) {};
                \path (l) ++(0:1) node[circle,fill=black,draw,inner sep=0pt,minimum size=4pt] (m) {};
                \path (m) ++(0:1) node[circle,fill=black,draw,inner sep=0pt,minimum size=4pt] (n) {};
                \path (n) ++(0:1) node[circle,fill=black,draw,inner sep=0pt,minimum size=4pt] (o) {};
                \path (o) ++(0:1) node[circle,fill=black,draw,inner sep=0pt,minimum size=4pt] (f) {};
                \path (f) ++(0:1) node[circle,fill=black,draw,inner sep=0pt,minimum size=4pt] (g) {};
                \path (g) ++(0:1) node[circle,fill=black,draw,inner sep=0pt,minimum size=4pt] (h) {};
         
        	\draw[thick] (a) to node[midway,inner sep=1pt,outer sep=1pt,minimum size=4pt,fill=white] {$\beta$} (b) to node[midway,inner sep=1pt,outer sep=1pt,minimum size=4pt,fill=white] {$\alpha$} (c) to node[midway,inner sep=1pt,outer sep=1pt,minimum size=4pt,fill=white] {$\beta$} (d) to node[midway,inner sep=1pt,outer sep=1pt,minimum size=4pt,fill=white] {$\alpha$} (e) to node[midway,inner sep=1pt,outer sep=1pt,minimum size=4pt,fill=white] {$\beta$} (l) to node[midway,inner sep=1pt,outer sep=1pt,minimum size=4pt,fill=white] {$\alpha$} (m)  (n) to node[midway,inner sep=1pt,outer sep=1pt,minimum size=4pt,fill=white] {$\beta$} (o) to node[midway,inner sep=1pt,outer sep=1pt,minimum size=4pt,fill=white] {$\alpha$} (f) to node[midway,inner sep=1pt,outer sep=1pt,minimum size=4pt,fill=white] {$\beta$} (g) to node[midway,inner sep=1pt,outer sep=1pt,minimum size=4pt,fill=white] {$\alpha$} (h);

                \draw[thick, dotted] (m) -- (n);

            \draw[decoration={brace,amplitude=10pt,mirror}, decorate] (0, -0.2) -- node [midway,below,xshift=0pt,yshift=-10pt] {$P'$} (6,-0.2);

            \node at (6.1, 0.3) {$x_{\ell' - 1}$};
            \node at (7.1, 0.3) {$x_{\ell'}$};
            \node at (0, 0.3) {$x$};
            \node at (1.1, 0.3) {$y_j$};

        \begin{scope}[yshift=2.9cm]
            \node[circle,fill=black,draw,inner sep=0pt,minimum size=4pt] (a) at (0,0) {};
        	\path (a) ++(0:1) node[circle,fill=black,draw,inner sep=0pt,minimum size=4pt] (b) {};
        	\path (b) ++(0:1) node[circle,fill=black,draw,inner sep=0pt,minimum size=4pt] (c) {};
                \path (c) ++(0:1) node[circle,fill=black,draw,inner sep=0pt,minimum size=4pt] (d) {};
                \path (d) ++(0:1) node[circle,fill=black,draw,inner sep=0pt,minimum size=4pt] (e) {};
                \path (e) ++(0:1) node[circle,fill=black,draw,inner sep=0pt,minimum size=4pt] (l) {};
                \path (l) ++(0:1) node[circle,fill=black,draw,inner sep=0pt,minimum size=4pt] (m) {};
                \path (m) ++(0:1) node[circle,fill=black,draw,inner sep=0pt,minimum size=4pt] (n) {};
                \path (n) ++(0:1) node[circle,fill=black,draw,inner sep=0pt,minimum size=4pt] (o) {};
                \path (o) ++(0:1) node[circle,fill=black,draw,inner sep=0pt,minimum size=4pt] (f) {};
                \path (f) ++(0:1) node[circle,fill=black,draw,inner sep=0pt,minimum size=4pt] (g) {};
                \path (g) ++(0:1) node[circle,fill=black,draw,inner sep=0pt,minimum size=4pt] (h) {};
         
        	\draw[thick] (a) to node[midway,inner sep=1pt,outer sep=1pt,minimum size=4pt,fill=white] {$\alpha$} (b) to node[midway,inner sep=1pt,outer sep=1pt,minimum size=4pt,fill=white] {$\beta$} (c) to node[midway,inner sep=1pt,outer sep=1pt,minimum size=4pt,fill=white] {$\alpha$} (d) to node[midway,inner sep=1pt,outer sep=1pt,minimum size=4pt,fill=white] {$\beta$} (e) to node[midway,inner sep=1pt,outer sep=1pt,minimum size=4pt,fill=white] {$\alpha$} (l) to node[midway,inner sep=1pt,outer sep=1pt,minimum size=4pt,fill=white] {$\beta$} (m) to node[midway,inner sep=1pt,outer sep=1pt,minimum size=4pt,fill=white] {$\alpha$} (n) to node[midway,inner sep=1pt,outer sep=1pt,minimum size=4pt,fill=white] {$\beta$} (o) to node[midway,inner sep=1pt,outer sep=1pt,minimum size=4pt,fill=white] {$\alpha$} (f) to node[midway,inner sep=1pt,outer sep=1pt,minimum size=4pt,fill=white] {$\beta$} (g) to node[midway,inner sep=1pt,outer sep=1pt,minimum size=4pt,fill=white] {$\alpha$} (h);

            \draw[decoration={brace,amplitude=10pt,mirror}, decorate] (0, -0.2) -- node [midway,below,xshift=0pt,yshift=-10pt] {$P$} (11,-0.2);
            \node at (0, 0.3) {$x$};
            \node at (1.1, 0.3) {$y_j$};

        \end{scope}

        \begin{scope}[yshift=2.1cm]
            \draw[-{Stealth[length=3mm,width=2mm]},very thick,decoration = {snake,pre length=3pt,post length=7pt,},decorate] (5.5,-0.5) -- (5.5,-1.5);
        \end{scope}
        	
        \end{tikzpicture}
		\caption{Flagging an edge.}\label{subfig:flag}
	\end{subfigure}%
	\qquad%
        \begin{subfigure}[t]{\textwidth}
		\centering
		\begin{tikzpicture}
            \clip (-0.5, 4.5) rectangle (12, -1.3);
		\node[circle,fill=black,draw,inner sep=0pt,minimum size=4pt] (a) at (0,0) {};
        	\path (a) ++(0:1) node[circle,fill=black,draw,inner sep=0pt,minimum size=4pt] (b) {};
        	\path (b) ++(0:1) node[circle,fill=black,draw,inner sep=0pt,minimum size=4pt] (c) {};
                \path (c) ++(0:1) node[circle,fill=black,draw,inner sep=0pt,minimum size=4pt] (d) {};
                \path (d) ++(0:1) node[circle,fill=black,draw,inner sep=0pt,minimum size=4pt] (e) {};
                \path (e) ++(0:5.5) node[circle,fill=black,draw,inner sep=0pt,minimum size=4pt] (f) {};
                \path (f) ++(0:1) node[circle,fill=black,draw,inner sep=0pt,minimum size=4pt] (g) {};
                \path (g) ++(0:1) node[circle,fill=black,draw,inner sep=0pt,minimum size=4pt] (h) {};
        	
        	\draw[thick, decorate,decoration=zigzag] (e) to (f);

            \path (a) ++(45:1) node[circle,fill=black,draw,inner sep=0pt,minimum size=4pt] (i) {};
            \path (a) ++(90:1) node[circle,fill=black,draw,inner sep=0pt,minimum size=4pt] (j) {};
            \path (a) ++(-45:1) node[circle,fill=black,draw,inner sep=0pt,minimum size=4pt] (k) {};
            \path (a) ++(-90:1) node[circle,fill=black,draw,inner sep=0pt,minimum size=4pt] (l) {};

            \node at (-0.3, -0.3) {$x_0$};
            \node at (11.6, -0.3) {$x_{\ell'-1}$};
            \node at (1, -0.3) {$x_1$};
            \node at (1.3, 0.7) {$y_{j-1}$};
        	
        	\draw[thick] (a) to node[midway,inner sep=1pt,outer sep=1pt,minimum size=4pt,fill=white] {$\beta$} (b) to node[midway,inner sep=1pt,outer sep=1pt,minimum size=4pt,fill=white] {$\alpha$} (c) to node[midway,inner sep=1pt,outer sep=1pt,minimum size=4pt,fill=white] {$\beta$} (d) to node[midway,inner sep=1pt,outer sep=1pt,minimum size=4pt,fill=white] {$\alpha$} (e) (f) to node[midway,inner sep=1pt,outer sep=1pt,minimum size=4pt,fill=white] {$\beta$} (g) to node[midway,inner sep=1pt,outer sep=1pt,minimum size=4pt,fill=white] {$\alpha$} (h) (a) to node[midway,inner sep=1pt,outer sep=1pt,minimum size=4pt,fill=white] {$\alpha$} (i);

                \draw[thick] (a) to node[midway,inner sep=1pt,outer sep=1pt,minimum size=4pt,fill=white] {$\xi$} (j) (k) to node[midway,inner sep=1pt,outer sep=1pt,minimum size=4pt,fill=white] {$\eta$} (a) to node[midway,inner sep=1pt,outer sep=1pt,minimum size=4pt,fill=white] {$\gamma$} (l);

        \begin{scope}[yshift=3cm]
            \node[circle,fill=black,draw,inner sep=0pt,minimum size=4pt] (a) at (0,0) {};
        	\path (a) ++(0:1) node[circle,fill=black,draw,inner sep=0pt,minimum size=4pt] (b) {};
        	\path (b) ++(0:1) node[circle,fill=black,draw,inner sep=0pt,minimum size=4pt] (c) {};
                \path (c) ++(0:1) node[circle,fill=black,draw,inner sep=0pt,minimum size=4pt] (d) {};
                \path (d) ++(0:1) node[circle,fill=black,draw,inner sep=0pt,minimum size=4pt] (e) {};
                \path (e) ++(0:5.5) node[circle,fill=black,draw,inner sep=0pt,minimum size=4pt] (f) {};
                \path (f) ++(0:1) node[circle,fill=black,draw,inner sep=0pt,minimum size=4pt] (g) {};
                \path (g) ++(0:1) node[circle,fill=black,draw,inner sep=0pt,minimum size=4pt] (h) {};
        	
        	\draw[thick, decorate,decoration=zigzag] (e) to (f);

            \path (a) ++(45:1) node[circle,fill=black,draw,inner sep=0pt,minimum size=4pt] (i) {};
            \path (a) ++(90:1) node[circle,fill=black,draw,inner sep=0pt,minimum size=4pt] (j) {};
            \path (a) ++(-45:1) node[circle,fill=black,draw,inner sep=0pt,minimum size=4pt] (k) {};
            \path (a) ++(-90:1) node[circle,fill=black,draw,inner sep=0pt,minimum size=4pt] (l) {};

            \node at (-0.3, -0.3) {$x_0$};
            \node at (11.6, -0.3) {$x_{\ell'-1}$};
            \node at (1, -0.3) {$x_1$};
            \node at (1.3, 0.7) {$y_{j-1}$};

            \draw[thick] (a) to node[midway,inner sep=1pt,outer sep=1pt,minimum size=4pt,fill=white] {$\xi$} (i) (k) to node[midway,inner sep=1pt,outer sep=1pt,minimum size=4pt,fill=white] {$\eta$} (a) to node[midway,inner sep=1pt,outer sep=1pt,minimum size=4pt,fill=white] {$\gamma$} (l);
            \draw[thick, dotted] (a) -- (j);
        	
        	\draw[thick] (a) to node[midway,inner sep=1pt,outer sep=1pt,minimum size=4pt,fill=white] {$\alpha$} (b) to node[midway,inner sep=1pt,outer sep=1pt,minimum size=4pt,fill=white] {$\beta$} (c) to node[midway,inner sep=1pt,outer sep=1pt,minimum size=4pt,fill=white] {$\alpha$} (d) to node[midway,inner sep=1pt,outer sep=1pt,minimum size=4pt,fill=white] {$\beta$} (e) (f) to node[midway,inner sep=1pt,outer sep=1pt,minimum size=4pt,fill=white] {$\alpha$} (g) to node[midway,inner sep=1pt,outer sep=1pt,minimum size=4pt,fill=white] {$\beta$} (h);

        \end{scope}

        \begin{scope}[yshift=2.5cm]
            \draw[-{Stealth[length=3mm,width=2mm]},very thick,decoration = {snake,pre length=3pt,post length=7pt,},decorate] (5.75,-0.5) -- (5.75,-1.5);
        \end{scope}
        	
		\end{tikzpicture}
		\caption{Augmenting when $x_{\ell' - 1} \neq y_{j-1}$.}\label{subfig:case1}
	\end{subfigure}%
	\qquad%
	\begin{subfigure}[b]{\textwidth}
		\centering
		\begin{tikzpicture}
            \clip (-7.8, 2.3) rectangle (5.1, -1.8);
            \begin{scope}
		\node[circle,fill=black,draw,inner sep=0pt,minimum size=4pt] (a) at (0,0) {};
        	\path (a) ++(0:1) node[circle,fill=black,draw,inner sep=0pt,minimum size=4pt] (b) {};
        	\path (b) ++(-30:1) node[circle,fill=black,draw,inner sep=0pt,minimum size=4pt] (c) {};
                \path (c) ++(-20:1) node[circle,fill=black,draw,inner sep=0pt,minimum size=4pt] (d) {};
                \path (d) ++(0:1) node[circle,fill=black,draw,inner sep=0pt,minimum size=4pt] (e) {};

            \path (a) ++(45:1) node[circle,fill=black,draw,inner sep=0pt,minimum size=4pt] (i) {};
            \path (a) ++(90:1) node[circle,fill=black,draw,inner sep=0pt,minimum size=4pt] (j) {};
            \path (a) ++(-45:1) node[circle,fill=black,draw,inner sep=0pt,minimum size=4pt] (k) {};
            \path (a) ++(-90:1) node[circle,fill=black,draw,inner sep=0pt,minimum size=4pt] (l) {};

            \path (i) ++(30:1) node[circle,fill=black,draw,inner sep=0pt,minimum size=4pt] (f) {};
            \path (f) ++(20:1) node[circle,fill=black,draw,inner sep=0pt,minimum size=4pt] (g) {};

            \node at (-0.3, -0.3) {$x_0$};
            \node at (1.2, 0.5) {$y_{j-1}$};
            \node at (1, -0.3) {$x_1$};
            \node at (-0.3, -1.3) {$y_{k-1}$};

            \draw[thick] (a) to node[midway,inner sep=1pt,outer sep=1pt,minimum size=4pt,fill=white] {$\eta$} (b) (k) to node[midway,inner sep=1pt,outer sep=1pt,minimum size=4pt,fill=white] {$\gamma$} (a) to node[midway,inner sep=1pt,outer sep=1pt,minimum size=4pt,fill=white] {$\xi$} (j);
            
            \draw[thick, decorate, decoration=zigzag] (g) to[out=10,in=10, looseness=2] (e);
        	
        	\draw[thick] (a) to node[midway,inner sep=1pt,outer sep=1pt,minimum size=4pt,fill=white] {$\beta$} (i) (b) to node[midway,inner sep=1pt,outer sep=1pt,minimum size=4pt,fill=white] {$\alpha$} (c) to node[midway,inner sep=1pt,outer sep=1pt,minimum size=4pt,fill=white] {$\beta$} (d) to node[midway,inner sep=1pt,outer sep=1pt,minimum size=4pt,fill=white] {$\alpha$} (e) (g) to node[midway,inner sep=1pt,outer sep=1pt,minimum size=4pt,fill=white] {$\beta$} (f) to node[midway,inner sep=1pt,outer sep=1pt,minimum size=4pt,fill=white] {$\alpha$} (i) (a) to node[midway,inner sep=1pt,outer sep=1pt,minimum size=4pt,fill=white] {$\alpha$} (l);

         \end{scope}

        \begin{scope}[xshift=-7cm]
            \node[circle,fill=black,draw,inner sep=0pt,minimum size=4pt] (a) at (0,0) {};
        	\path (a) ++(0:1) node[circle,fill=black,draw,inner sep=0pt,minimum size=4pt] (b) {};
        	\path (b) ++(-30:1) node[circle,fill=black,draw,inner sep=0pt,minimum size=4pt] (c) {};
                \path (c) ++(-20:1) node[circle,fill=black,draw,inner sep=0pt,minimum size=4pt] (d) {};
                \path (d) ++(0:1) node[circle,fill=black,draw,inner sep=0pt,minimum size=4pt] (e) {};

            \path (a) ++(45:1) node[circle,fill=black,draw,inner sep=0pt,minimum size=4pt] (i) {};
            \path (a) ++(90:1) node[circle,fill=black,draw,inner sep=0pt,minimum size=4pt] (j) {};
            \path (a) ++(-45:1) node[circle,fill=black,draw,inner sep=0pt,minimum size=4pt] (k) {};
            \path (a) ++(-90:1) node[circle,fill=black,draw,inner sep=0pt,minimum size=4pt] (l) {};

            \path (i) ++(30:1) node[circle,fill=black,draw,inner sep=0pt,minimum size=4pt] (f) {};
            \path (f) ++(20:1) node[circle,fill=black,draw,inner sep=0pt,minimum size=4pt] (g) {};

            \node at (-0.3, -0.3) {$x_0$};
            \node at (1.2, 0.5) {$y_{j-1}$};
            \node at (1, -0.3) {$x_1$};
            \node at (-0.3, -1.3) {$y_{k-1}$};

            \draw[thick] (a) to node[midway,inner sep=1pt,outer sep=1pt,minimum size=4pt,fill=white] {$\xi$} (i) (k) to node[midway,inner sep=1pt,outer sep=1pt,minimum size=4pt,fill=white] {$\eta$} (a) to node[midway,inner sep=1pt,outer sep=1pt,minimum size=4pt,fill=white] {$\gamma$} (l);
            \draw[thick, dotted] (a) -- (j);
            
            \draw[thick, decorate, decoration=zigzag] (g) to[out=10,in=10, looseness=2] (e);
        	
        	\draw[thick] (a) to node[midway,inner sep=1pt,outer sep=1pt,minimum size=4pt,fill=white] {$\alpha$} (b) to node[midway,inner sep=1pt,outer sep=1pt,minimum size=4pt,fill=white] {$\beta$} (c) to node[midway,inner sep=1pt,outer sep=1pt,minimum size=4pt,fill=white] {$\alpha$} (d) to node[midway,inner sep=1pt,outer sep=1pt,minimum size=4pt,fill=white] {$\beta$} (e) (g) to node[midway,inner sep=1pt,outer sep=1pt,minimum size=4pt,fill=white] {$\alpha$} (f) to node[midway,inner sep=1pt,outer sep=1pt,minimum size=4pt,fill=white] {$\beta$} (i);
        \end{scope}

        \draw[-{Stealth[length=3mm,width=2mm]},very thick,decoration = {snake,pre length=3pt,post length=7pt,},decorate] (-1.5,0.3) -- (-0.5,0.3);
		\end{tikzpicture}
		\caption{Augmenting when $x_{\ell' - 1} = y_{j-1}$.}\label{subfig:case2}
	\end{subfigure}
	\caption{Steps~\ref{step:start}--\ref{step:end} of Algorithm~\ref{alg:color}.}\label{fig:cases}
\end{figure}

Let us now describe the second stage.
Let $\phi$ be the partial coloring from the first stage.
Let $G^*$ be the subgraph induced by the uncolored edges.
(Note that these are precisely the edges which were flagged during the first stage.)
If $\Delta(G^*) > \eps\Delta/6$, we return $\mathsf{FAIL}$.
If not, we properly color $G^*$ with $3\Delta(G^*)$ colors using Algorithm~\ref{alg:greedy}.
Let $\psi$ be the resulting coloring on $G^*$.
For each edge $e \in E(G^*)$, we let $\phi(e) = (1+\epsilon/2)\Delta + \psi(e)$.
Since $3\Delta(G^*) \leq \epsilon\Delta/2$, this defines a proper $(1+\epsilon)\Delta$-edge-coloring of $G$.

\vspace{5pt}
\begin{breakablealgorithm}\small
\caption{$(1+\eps)\Delta$-Edge-Coloring}\label{alg:color}
\begin{flushleft}
\textbf{Input}: An $n$-vertex graph $G$ having $m$ edges and maximum degree $\Delta$, and two parameters $\kappa$ and $\ell$. \\
\textbf{Output}: $\mathsf{FAIL}$ or a proper $(1 + \epsilon)\Delta$-edge-coloring of $G$.
\end{flushleft}

\begin{algorithmic}[1]
    \State $U \gets E$, \quad $\phi(e) \gets \blank$ \textbf{for each} $e \in U$
    \While{$U \neq \0$}
        \State Pick an edge $e \in U$ and a vertex $x \in e$ uniformly at random. \label{step:random_edge}
        \State Define $C$ by sampling with replacement $\kappa$ times from $[(1+\epsilon/2)\Delta]$. \label{step:random_sample}
        \State $\mathsf{Out} \gets \hyperref[alg:viz]{\mathsf{VizingChain}}(\phi, e, x, C, \ell)$ 
        \medskip\Comment{Algorithm~\ref{alg:viz}} \label{step:call_to_viz}
        \If{$\mathsf{Out} = \mathsf{FAIL}$}
            \State \Return $\mathsf{FAIL}$ \label{step:fail_stage1}
        \EndIf
        \medskip
        \State $((F, P), \alpha) \gets \mathsf{Out}$
        \State Let $F = (x, y_0, \ldots, y_{k-1})$ and let $P = (x_0 = x, x_1, \ldots, x_s)$.
        \If{$s = 0$}
            \State Shift the fan $F$ and color the edge $xy_{k-1}$ with $\alpha$.
            \State \textbf{continue}
        \ElsIf{$s = \ell$} \label{step:start}
            \State Let $\ell' \in [\ell]$ be an integer chosen uniformly at random. \label{step:random_choice}
            \State Uncolor the edge $x_{\ell'-1}x_{\ell'}$. \label{step:uncolor_edge}
            \State $P' \gets (x_0, \ldots, x_{\ell'-1})$.
        \Else
            \State $P ' \gets P$.
        \EndIf
        \medskip
        \State Let $j$ be such that $y_j = x_1$ and let $F' = (x, y_0, \ldots, y_{j-1})$.
        \State Flip the path $P'$.
        \If{$x_{\ell' - 1} = y_{j-1}$}
            \State Shift the fan $F$ and color the edge $xy_{k-1}$ with $\alpha$.
        \Else
            \State Shift the fan $F'$ and color the edge $xy_{j-1}$ with $\alpha$.
        \EndIf \label{step:end}
        \medskip
        \State $U \gets U \setminus \set{e}$
    \EndWhile
    \medskip
    \State Let $G^*$ be the subgraph induced by the uncolored edges. \label{step:define_G_star}
    \If{$\Delta(G^*) > \epsilon\,\Delta / 6$} 
        \State \Return $\mathsf{FAIL}$ \label{step:fail_stage2}
    \EndIf
    \State Run Algorithm~\ref{alg:greedy} to obtain a $3\Delta(G^*)$-edge-coloring $\psi$ of $G^*$ \label{step:sinnamon}
    \State $\phi(e) \gets \psi(e) + \lfloor(1 + \epsilon/2)\Delta\rfloor$ for each $e \in E(G^*)$ \label{step:combine}
    \State \Return $\phi$
\end{algorithmic}
\end{breakablealgorithm}
\vspace{5pt}

A few remarks are in order.
First, note that as a result of the way we manage the set $U$, we need not explicitly ``flag'' edges at Step~\ref{step:uncolor_edge}.
From our earlier analysis of Algorithm~\ref{alg:viz} and since $|C| \leq \kappa$, Step~\ref{step:call_to_viz} takes $O(\kappa^2 + \ell)$ time.
Additionally, the shifting and flipping operations take time proportional to the lengths of the fans and paths involved, and so we may conclude that each iteration takes $O(\kappa^2 + \ell)$ time.
It follows that the first stage takes time $O(m\,(\kappa^2 + \ell))$.
The verification at Step~\ref{step:fail_stage2} can be performed in $O(m)$ time and as a result of Proposition~\ref{prop:greedy}, Step~\ref{step:sinnamon} takes $O(\max\set{m,\,\log^2 n})$ time with probability at least $1-1/\poly(n)$.
Finally, Step~\ref{step:combine} takes $O(m)$ time.
In total, the runtime of Algorithm~\ref{alg:color} is
\begin{align}\label{eqn:runtime}
    O\left(\max\left\{m\,\kappa^2,\,m\,\ell,\,m,\,\log^2n\right\}\right),
\end{align}
with probability at least $1-1/\poly(n)$.

\subsection{Proof of Correctness}\label{subsec:poc}

In this section, we will prove the correctness of our algorithm.
Note that as a result of Proposition~\ref{prop:greedy} and Step~\ref{step:fail_stage2}, it is enough to show that Stage 1 constructs a proper partial $(1 + \eps/2)\Delta$-edge-coloring.
Let us first consider the output of Algorithm~\ref{alg:fan}.

\begin{lemma}\label{lemma:fan}
    Consider running Algorithm~\ref{alg:fan} on input $(\phi, xy, x, C)$.
    Suppose the algorithm does not fail and outputs $(F, \alpha, j)$ where $F = (x, y_0 = y, y_1, \ldots, y_{k-1})$.
    Let $\beta \in M(\phi, x)$ be arbitrary and let $P$ be the maximal $\alpha\beta$-path under $\phi$ starting at $x$.
    One of the following must hold:
    \begin{enumerate}
        \item either $j = k$ and the edge $xy_{k-1}$ is $\psi$-happy, where $\psi$ is the coloring obtained from $\phi$ by shifting $F$, or
        \item the edge $xy_{j-1}$ is $\psi$-happy, where $\psi$ is the coloring obtained from $\phi$ by flipping $P$ and then shifting the fan $F' = (x, y_0 = y, y_1, \ldots, y_{j-1})$, or
        \item the edge $xy_{k-1}$ is $\psi$-happy, where $\psi$ is the coloring obtained from $\phi$ by flipping $P$ and then shifting $F$.
    \end{enumerate}
    Moreover, the happy edge can be colored with the color $\alpha$.
\end{lemma}

\begin{proof}
    If $j = k$, we are at Step~\ref{step:happy_fan}. 
    In particular, $\alpha \in M(\phi, x) \cap M(\phi, y_{k-1})$.
    Note that $M(\psi, x) = M(\phi, x)$ and $M(\psi, y_{k-1}) \supseteq M(\phi, y_{k-1})$.
    It follows that $\alpha \in M(\psi, x) \cap M(\psi, y_{k-1})$ as well, implying the edge $xy_{k-1}$ is $\psi$-happy and may be colored $\alpha$.

    Now, suppose that we do not reach Step~\ref{step:happy_fan}.
    In this case we must have $1 < j < k$.
    Let $\psi'$ be the coloring obtained from $\phi$ by flipping $P$.
    Note that we have $\alpha \in M(\psi', x)$.
    Furthermore, by construction we have $\alpha \in M(\phi, y_{j-1}) \cap M(\phi, y_{k-1})$.
    Additionally, for all vertices $v \in V\setminus\set{x,\, \vend(P)}$, we have $M(\psi', v) = M(\phi, v)$.    
    If $\vend(P) \neq y_{j-1}$ (i.e., $\alpha \in M(\psi', y_{j-1})$), then $F'$ is a valid fan under the coloring $\psi'$ and an identical argument to the previous paragraph shows that $xy_{j-1}$ is $\psi$-happy, where $\psi$ is obtained from $\psi'$ by shifting $F'$.
    If not, then it must be the case that $\alpha \in M(\psi', y_{k-1})$ and $\beta \in M(\psi', y_{j-1})$.
    It follows similarly that $F$ is a valid fan under the coloring $\psi'$, and $xy_{k-1}$ is $\psi$-happy, where $\psi$ is obtained from $\psi'$ by shifting $F$.
\end{proof}

Now, let us consider the output of Algorithm~\ref{alg:viz}.

\begin{lemma}\label{lemma:viz}
    Consider running Algorithm~\ref{alg:viz} on input $(\phi, xy, x, C, \ell)$.
    Suppose the algorithm does not fail and outputs $((F, P), \alpha)$ such that $F = (x, y_0=y, \ldots, y_{k-1})$ and $P = (x_0 = x, \ldots, x_s)$ is an $\alpha\beta$-path.
    One of the following must hold for $j$ such that $y_j = x_1$ (if $s\geq 1$):
    \begin{enumerate}
        \item either $P = (x)$ and the edge $xy_{k-1}$ is $\psi$-happy, where $\psi$ is the coloring obtained from $\phi$ by shifting $F$, or
        \item $\length(P) < \ell$ and the edge $xy_{j-1}$ is $\psi$-happy, where $\psi$ is the coloring obtained from $\phi$ by flipping $P$ and then shifting $F' = (x, y_0 = y, y_1, \ldots, y_{j-1})$, or
        \item $\length(P) < \ell$ and the edge $xy_{k-1}$ is $\psi$-happy, where $\psi$ is the coloring obtained from $\phi$ by flipping $P$ and then shifting $F$, or
        \item $\length(P) = \ell$.
    \end{enumerate}
    Moreover, the happy edge may be colored with the color $\alpha$.
\end{lemma}

\begin{proof}
    The proof follows by construction and by Lemma~\ref{lemma:fan}.
\end{proof}

Finally, let us show that the first stage computes a proper partial $(1+\epsilon/2)\Delta$-edge-coloring, completing the proof of correctness.

\begin{lemma}\label{lemma:stage1}
    Consider running Algorithm~\ref{alg:color} on input $(G, \kappa, \ell)$.
    Suppose the algorithm does not fail in the \textsf{\upshape{while}} loop and let $\phi$ be the resulting coloring at the end of the loop.
    Then, $\phi$ is a proper partial $(1+\epsilon/2)\Delta$-edge-coloring of $G$.
\end{lemma}

\begin{proof}
    It is enough to show that the partial coloring at the end of each iteration is proper.
    As a result of Lemma~\ref{lemma:viz}, the only case to consider is when $\length(P) = \ell$.
    Note that by uncoloring the edge $x_{\ell'-1}x_{\ell'}$ at Step~\ref{step:uncolor_edge}, the path $P'$ is a maximal $\alpha\beta$-path.
    The claim now follows by an identical argument as in the proof of Lemma~\ref{lemma:fan}.
\end{proof}

\subsection{Probabilistic Analysis}\label{sec:proof}

Let us fix $\kappa = \Theta(\log n / \epsilon)$ and $\ell = \Theta(\kappa^2\log n/\epsilon)$, where the hidden constant factors are at least $50$.
In particular, as a result of \eqref{eqn:runtime}, the runtime of Algorithm~\ref{alg:color} is $O(m\log^3n/\epsilon^3)$ as claimed.
All that remains is to bound the probability of failure.
To this end, we make the following definitions for each of the $m$ iterations of the \textsf{while} loop in Algorithm~\ref{alg:color}:
\begin{align*}
    \phi_i &\defeq \text{ the coloring at the start of the $i$-th iteration}, \\
    e_i = x_iy_i &\defeq \text{ the random edge and vertex picked at Step~\ref{step:random_edge}}, \\
    C_i &\defeq \text{ the random subset of colors sampled at Step~\ref{step:random_sample}}, \\
    \ell_i &\defeq \text{ the random choice made at Step~\ref{step:random_choice}}, \\
    f_i &\defeq \text{ the edge uncolored at Step~\ref{step:uncolor_edge}}.
\end{align*}
Furthermore, let $G^*$ be the subgraph at Step~\ref{step:define_G_star} in Algorithm~\ref{alg:color}.

The algorithm fails if we either reach Step~\ref{step:fail_stage1} or Step~\ref{step:fail_stage2}.
Let $\mathcal{F}_1$ and $\mathcal{F}_2$ denote these events.
Note that we reach Step~\ref{step:fail_stage1} at iteration $i$ if we either reach Step~\ref{step:fail_fan} at some iteration of the \textsf{while} loop in Algorithm~\ref{alg:fan} or we reach Step~\ref{step:fail_x_alpha} in Algorithm~\ref{alg:viz}.
In particular, if $\mathcal{F}_1$ occurs, there is some iteration $i$ and vertex $v$ such that $M(\phi_i, v) \cap C_i = \0$.
Let $\tilde{\mathcal{F}}_1$ denote this event.
We have
\[\Pr[\text{Failure}] \,=\, \Pr[\mathcal{F}_1 \cup \mathcal{F}_2] \,\leq\, \Pr[\tilde{\mathcal{F}}_1 \cup \mathcal{F}_2] \,\leq\, \Pr[\tilde{\mathcal{F}}_1] + \Pr[\mathcal{F}_2].\]

The following lemma will assist with bounding $\Pr[\tilde{\mathcal{F}}_1]$.

\begin{lemma}\label{lemma:prob_C_i}
    $\Pr[\exists v \in V,\,M(\phi_i, v) \cap C_i = \0 \mid \phi_i] \leq 1/\poly(n)$.
\end{lemma}

\begin{proof}
    Consider the $i$-th iteration and fix a vertex $v$.
    As we sample with replacement while defining $C_i$, we have:
    \[\Pr[M(\phi_i, v) \cap C_i = \0 \mid \phi_i] = \left(1 - \frac{|M(\phi_i, v)|}{(1+\epsilon/2)\Delta}\right)^\kappa \leq \exp\left(- \frac{|M(\phi_i, v)|\kappa}{(1+\epsilon/2)\Delta}\right).\]
    Note that $|M(\phi_i, v)| \geq \epsilon\Delta/2$ for each vertex $v$.
    Therefore,
    \[\Pr[M(\phi_i, v) \cap C_i = \0 \mid \phi_i] \,\leq\, \exp\left(- \frac{\eps\kappa}{2+\epsilon}\right) \,=\, \frac{1}{\poly(n)},\]
    where the last step follows by definition of $\kappa$ and since $\eps < 1$.
    The claim now follows by a union bound over $V$.
\end{proof}

As the bound in Lemma~\ref{lemma:prob_C_i} is independent of $i$, we have
\begin{align*}
    \Pr[\exists v \in V,\,M(\phi_i, v) \cap C_i = \0] &= \sum_{\phi}\Pr[\phi_i = \phi]\Pr[\exists v \in V,\,M(\phi_i, v) \cap C_i = \0\mid \phi_i = \phi] \\
    &\leq \frac{1}{\poly(n)},
\end{align*}
where the sum is taken over all possible such colorings $\phi$.
Taking a union bound over all iterations, we have
\begin{align}\label{eqn:prob_stage1}
    \Pr[\tilde{\mathcal{F}}_1] \,\leq\, \frac{m}{\poly(n)} \,=\, \frac{1}{\poly(n)}.
\end{align}

Let us now consider $\Pr[\mathcal{F}_2]$.
We will define the following random variables for each vertex $v \in V(G)$:
\[d_i(v) \defeq \bbone{v \in f_i}, \quad d(v) \defeq \deg_{G^*}(v) = \sum_{i = 1}^md_i(v).\]
In the following lemma, we shall bound $\Pr[d_i(v) = 1\mid \phi_i,\, C_i]$.

\begin{lemma}\label{lemma:main_lemma}
    $\Pr[d_i(v) = 1\mid \phi_i,\, C_i] \leq \dfrac{\eps\Delta}{50(m - i + 1)\log n}$.
\end{lemma}

\begin{proof}
    For $v$ to be in $f_i$, there must be an $\alpha\beta$-path $P$ passing through $v$ for some $\alpha, \beta \in C_i$ such that one of the endpoints of $P$ is $x_i$.
    Let $X$ be the set of all vertices that are endpoints of some such alternating path $P$.
    We have
    \[|X| \leq 2\,\binom{|C_i|}{2} \leq \kappa^2.\]
    For a fixed $x_i$, the vertex $y_i$ must be adjacent to $x_i$ and so there are at most $\Delta$ such vertices.
    Note that the random choices made are $e_i$, $x_i$, and $\ell_i$ (as we are conditioning on the outcome $C_i$).
    As there are at most two edges incident to $v$ in any alternating path $P$ passing through $v$, we have
    \begin{align*}
        \Pr[d_i(v) = 1\mid \phi_i,\, C_i] &\leq \sum_{x \in X}\Pr[x_i = x, \, v\in f_i \mid \phi_i,\, C_i] \\
        &\leq \sum_{x \in X}\frac{\Delta}{2(m - i + 1)}\,\Pr[v \in f_i \mid \phi_i,\, C_i,\, x_i = x] \\
        &\leq \sum_{x \in X}\frac{\Delta}{(m - i + 1)\ell} \\
        &= \frac{\Delta|X|}{\ell\,(m - i + 1)},
    \end{align*}
    where we use the fact that the size of the set $U$ during the $i$-th iteration is $m - i + 1$.
    The claim now follows since $|X| \leq \kappa^2$ and by the definition of $\ell$.
\end{proof}

Note that $d_i(v)$ is independent of all $d_j(v)$ for $j < i$, given $\phi_i$ and $C_i$.
In particular, as a result of Lemma~\ref{lemma:main_lemma} we have
\begin{align*}
    \E[d_i(v) \mid d_1(v), \ldots, d_{i - 1}(v)] &= \E[\E[d_i(v) \mid d_1(v), \ldots, d_{i - 1}(v), \phi_i,\,C_i]] \\
    &= \E[\E[d_i(v) \mid \phi_i,\,C_i]] \\
    &= \E[\Pr[d_i(v) = 1 \mid \phi_i,\,C_i]] \\
    &\leq \frac{\eps\Delta}{50(m - i + 1)\log n}.
\end{align*}

To finish the analysis of $d(v)$, we shall apply the following concentration inequality due to Kuszmaul and Qi \cite{azuma}, which is a special case of their version of multiplicative Azuma's inequality for supermartingales:

\begin{theorem}[{Kuszmaul--Qi \cite[Corollary 6]{azuma}}]\label{theo:azuma_supermartingale}
    Let $c > 0$ and let $X_1$, \ldots, $X_n$ be 
    random variables taking values in $[0,c]$.
    Suppose that $\E[X_i|X_1, \ldots, X_{i-1}] \leq a_i$ for all $i$.
    Let $\mu \defeq \sum_{i = 1}^na_i$. Then, for any $\delta > 0$,
    \[\Pr\left[\sum_{i = 1}^nX_i \geq (1+\delta)\mu\right] \,\leq\, \exp\left(-\frac{\delta^2\mu}{(2+\delta)c}\right).\]
\end{theorem}

From our earlier computation, we may apply Theorem~\ref{theo:azuma_supermartingale} with 
\[X_i = d_i(v), \quad c = 1, \quad \text{and} \quad a_i = \frac{\eps\Delta}{50(m - i + 1)\log n}.\]
Assuming $G$ has no isolated vertices, we must have $n/2 \leq m \leq n^2$.
With this in hand, from the bounds on partial sums of the harmonic series \ep{see, e.g., \cite[\S2.4]{bressoud2022radical}}, we have
\[\frac{\eps\Delta}{50} \,\leq\, \mu \,\leq\, \frac{\eps\Delta}{12}.\]
Setting $\delta = 1$, we conclude
\[\Pr\left[d(v) \geq \frac{\eps\Delta}{6}\right] \,\leq\, \exp\left(-\frac{\eps\Delta}{150}\right) \,=\, \frac{1}{\poly(n)},\]
where the last step follows by the lower bound on $\Delta$ stated in Theorem~\ref{theo:main_theo}.
By taking a union bound over $V(G)$, we have
\[\Pr[\mathcal{F}_2] \leq \frac{1}{\poly(n)}.\]
Putting this together with \eqref{eqn:prob_stage1}, we have
\[\Pr[\text{Failure}] \leq \frac{1}{\poly(n)},\]
as desired.

\sloppy
\printbibliography

\appendix

\section{Proof of Proposition~\ref{prop:greedy}}\label{appendix}

Let us first restate the proposition.

\begin{proposition*}[{Restatement of Proposition~\ref{prop:greedy}}]
    Let $\eps > 0$ be arbitrary and let $\gamma \defeq \min\set{1,\, \eps}$.
    For any $n$-vertex graph $G$ with $m$ edges and maximum degree $\Delta$,
    Algorithm~\ref{alg:greedy} computes a proper $(2+\eps)\Delta$-edge-coloring of $G$ in $O(\max\set{m/\gamma^2,\,\log^2n/\gamma})$ time with probability at least $1 - \min\{e^{-m/\gamma},\\ 1/\poly(n)\}$.
\end{proposition*}

Let $\phi_i$ denote the coloring at the start of the $i$-th iteration and let $T_i$ denote the number of random choices made during the $i$-th iteration.
Note that the runtime of Algorithm~\ref{alg:greedy} is precisely $O(T)$, where $T = \sum_{i = 1}^mT_i$.
Let $xy$ be the uncolored edge to be colored during the $i$-th iteration.
We have
\[|M(\phi_i, x)|, |M(\phi_i, y)| \geq (2+\eps)\Delta - (\Delta - 1) = 1 + (1+\eps)\Delta,\]
which implies
\begin{align*}
    |M(\phi_i, x) \cap M(\phi_i, y)| &= |M(\phi_i, x)| + |M(\phi_i, y)| - |M(\phi_i, x) \cup M(\phi_i, y)| \\
    &\geq 2(1 + (1+\eps)\Delta) - (2+\eps)\Delta \\
    &\geq 2 + \eps\Delta.
\end{align*}
In particular, 
\[\Pr[T_i \geq t_i \mid \phi_i] \leq \left(1 - \frac{(2 + \eps\Delta)}{(2+\eps)\Delta}\right)^{t_i - 1} \leq \left(1 - \frac{\gamma}{3}\right)^{t_i - 1},\]
where $\gamma$ is as defined in the statement of Proposition~\ref{prop:greedy}.
First, let us consider the case that $m = O(\log n)$.
As the expression above is independent of $\phi_i$, we may conclude the following:
\begin{align*}
    \Pr[T \geq t] \,\leq\, \Pr[\exists i,\, T_i \geq t/m] \,\leq\, m\left(1 - \frac{\gamma}{3}\right)^{t/m - 1} \,\leq\, m\exp\left(-\frac{\gamma}{3}\left(\frac{t}{m} - 1\right)\right).
\end{align*}
For $t = \Theta(\log^2 n/\gamma)$, the above is at most $1/\poly(n)$, as desired.

Now, suppose $m = \Omega(\log n)$.
Let $X_1, \ldots, X_m$ be i.i.d. Geometric random variables with probability of success $\gamma/3$.
Clearly, $X = \sum_{i = 1}^mX_i$ stochastically dominates $T$.
The following result of \cite{bernshteyn2024linear} which follows as a corollary to \cite[Proposition 2.2]{assadi2024faster} will assist with our proof.

\begin{proposition}[{\cite[Proposition~7.3]{bernshteyn2024linear}}]\label{prop:geom}
    Let $X_1, \ldots, X_n$ be i.i.d geometric random variables with probability of success $p$.
    Then for any $\mu \geq n/p$,
    \[\Pr\left[\sum_{i = 1}^nX_i \,\geq\, \frac{6\mu}{p(1-p)}\right] \leq \exp\left(-\frac{\mu}{1-p}\right).\]
\end{proposition}

Applying Proposition~\ref{prop:geom} to $X_1$, \ldots, $X_m$ with $p = \gamma/3$ and $\mu = m/p$, we have
\[\Pr\left[X \,\geq\, \frac{90m}{\gamma^2}\right] \,\leq\, \Pr\left[X \,\geq\, \frac{6m}{p^2(1 - p)}\right] \,\leq\, \exp\left(-\frac{m}{p(1-p)}\right) \,\leq\, \exp\left(-\frac{m}{\gamma}\right),\]
as desired.

\end{document}